\documentclass{LMCS}

\def\doi{8 (2:12) 2012}
\lmcsheading%
{\doi}
{1--27}
{}
{}
{Jan.~\phantom07, 2011}
{Jun.~19, 2012}
{}

\usepackage{enumerate}
\usepackage{hyperref}
\usepackage{amsmath}
\usepackage{verbatim}
\usepackage{neil}
\usepackage{xypic}
\usepackage[all,2cell]{xy} \UseAllTwocells
\usepackage{amssymb}

\newcommand{\A}{\mathcal{A}}
\newcommand{\E}{\mathcal{E}}
\newcommand{\B}{\mathcal{B}}
\newcommand{\Set}{\mbox{{\sf Set}}}
\newcommand{\Nat}{\mathit{Nat}}
\newcommand{\Alge}[1]{\mathit{Alg}_{#1}}
\newcommand{\hash}{\#}

\newcommand{\adjunction}[4]{\xymatrix{ {#1} \ar@/^/[r]^{#3}
    \ar@{}[r]|{\top}& {#2} \ar@/^/[l]^{#4}}}

\newcommand\dia[1]{\vcenter{\xymatrix{#1}}}

\begin{document}

\title[Generic Fibrational Induction]{Generic Fibrational Induction\rsuper*}

\author[N.~Ghani]{Neil Ghani}	
\author[P.~Johann]{Patricia Johann}	
\author[C.~Fumex]{Cl\'ement Fumex}	
\address{University of Strathclyde, Glasgow G1 1XH, UK}	
\email{\{Neil.Ghani, Patricia.Johann, Clement.Fumex\}@cis.strath.ac.uk} 
\thanks{This research is partially supported by EPSRC grant EP/C0608917/1.}

\titlecomment{{\lsuper*}This paper
  is revised and significantly expanded version of~\cite{gjf10}.}


\begin{abstract}
  \noindent This paper provides an induction rule that can be used to
  prove properties of data structures whose types are inductive, i.e.,
  are carriers of initial algebras of functors. Our results are
  semantic in nature and are inspired by Hermida and Jacobs' elegant
  algebraic formulation of induction for polynomial data types.  Our
  contribution is to derive, under slightly different assumptions, a
  sound induction rule that is generic over {\em all} inductive types,
  polynomial or not. Our induction rule is generic over the kinds of
  properties to be proved as well: like Hermida and Jacobs, we work in
  a general fibrational setting and so can accommodate very general
  notions of properties on inductive types rather than just those of a
  particular syntactic form. We establish the soundness of our generic
  induction rule by reducing induction to iteration. We then show how
  our generic induction rule can be instantiated to give induction
  rules for the data types of rose trees, finite hereditary sets, and
  hyperfunctions.  The first of these lies outside the scope of
  Hermida and Jacobs' work because it is not polynomial, and as far as
  we are aware, no induction rules have been known to exist for the
  second and third in a general fibrational framework. Our
  instantiation for hyperfunctions underscores the value of working in
  the general fibrational setting since this data type cannot be
  interpreted as a set.
\end{abstract}

\maketitle

\section{Introduction}
Iteration operators provide a uniform way to express common and
naturally occurring patterns of recursion over inductive data types.
Expressing recursion via iteration operators makes code easier to
read, write, and understand; facilitates code reuse; guarantees
properties of programs such as totality and termination; and supports
optimising program transformations such as $\mathit{fold}$ fusion and
short cut fusion. Categorically, iteration operators arise from
initial algebra semantics of data types, in which data types are
regarded as carriers of initial algebras of functors. Lambek's Lemma
ensures that the carrier of the initial algebra of $F$ is its least
fixed point $\mu F$, and initiality ensures that, given any
$F$-algebra $h:F A \rightarrow A$, there is a unique $F$-algebra
homomorphism, denoted $fold\, h$, from the initial algebra
$\mathit{in} : F(\mu F) \ra \mu F$ to that algebra.  For each functor
$F$, the map $\mathit{fold} : (F A \ra A) \ra \mu F \ra A$ is the
iteration operator for the data type $\mu F$.  Initial algebra
semantics thus provides a well-developed theory of
iteration which is... \medskip
\begin{iteMize}{$\bullet$}
\item {\em ...principled}, in that it is derived solely from the
  initial algebra semantics of data types. This is important because
  it helps ensure that programs have rigorous mathematical foundations
  that can be used to ascertain their meaning and correctness. \medskip
\item {\em ...expressive}, and so is applicable to all inductive types
  --- i.e., to {\em every} type which is the carrier of an initial
  algebra of a functor --- rather than just to syntactically defined
  classes of data types such as polynomial data types. \medskip
\item {\em ...correct}, and so is valid in any model ---
  set-theoretic, domain-theoretic, realisability, etc. --- in which
  data types are interpreted as carriers of initial algebras. \medskip
\end{iteMize}

\noindent Because induction and iteration are closely linked --- induction is
often used to prove properties of functions defined by iteration, and
the correctness of induction rules is often established by reducing it
to that of iteration --- we may reasonably expect that initial algebra
semantics can be used to derive a principled, expressive, and correct
theory of induction for data types as well. In most treatments of
induction, given a functor $F$ together with a property $P$ to be
proved about data of type $\mu F$, the premises of the induction rule
for $\mu F$ constitute an $F$-algebra with carrier $\Sigma x: \mu F.\,
P x$. The conclusion of the rule is obtained by supplying such an
$F$-algebra as input to the iteration operator for $\mu F$. This
yields a function from $\mu F$ to $\Sigma x: \mu F.\, P x$ from which
a function of type $\forall x:\mu F. \,P x$ can be obtained. It has
not, however, been possible to characterise $F$-algebras with carrier
$\Sigma x:\mu F. \,P x$ without additional assumptions on
$F$. Induction rules are thus typically derived under the assumption
that the functors involved have a certain structure, e.g., that they
are polynomial. Moreover, taking the carriers of the algebras to be
$\Sigma$-types assumes that properties are represented as type-valued
functions. So while induction rules derived as described above are
both principled and correct, their expressiveness is limited along two
dimensions: with respect to the data types for which they can be
derived and the nature of the properties they can verify.

A more expressive, yet still principled and correct, approach to
induction is given by Hermida and Jacobs~\cite{hj98}. They show how to
lift each functor $F$ on a base category of types to a functor
$\hat{F}$ on a category of properties over those types, and take the
premises of the induction rule for the type $\mu F$ to be an
$\hat{F}$-algebra. Hermida and Jacobs work in a fibrational setting
and the notion of property they consider is, accordingly, very
general. Indeed, they accommodate any notion of property that can be
suitably fibred over the category of types, and so overcome one of the
two limitations mentioned above. On the other hand, their approach
gives sound induction rules only for polynomial data types, so the
limitation on the class of data types treated remains in their work.

This paper shows how to remove the restriction on the class of data
types treated.  {\em Our main result is a derivation of a sound
  generic induction rule that can be instantiated to {\em every}
  inductive type}, regardless of whether it is polynomial or not. We
think this is important because it provides a counterpart for
induction to the existence of an iteration operator for every
inductive type.  We take Hermida and Jacobs' approach as our point of
departure and show that, under slightly different assumptions on the
fibration involved, we can lift {\em any} functor on the base category
of a fibration to a functor on the total category of the
fibration. The lifting we define forms the basis of our generic
induction rule.

The derivation of a generic, sound induction rule covering all
inductive types is clearly an important theoretical result, but it
also has practical consequences: \smallskip
\begin{iteMize}{$\bullet$}
\item We show in Example~\ref{ex:rose} how our generic induction rule
  can be instantiated to the families fibration over $\Set$ (the
  fibration most often implicitly used by type theorists and those
  constructing inductive proofs with theorem provers) to derive the
  induction rule for rose trees that one would intuitively expect. The
  data type of rose trees lies outside the scope of Hermida and
  Jacobs' results because it is not polynomial. On the other hand, an
  induction rule for rose trees is available in the proof assistant
  Coq, although it is neither the one we intuitively expect nor
  expressive enough to prove properties that ought to be amenable to
  inductive proof. Indeed, if we define rose trees in Coq by\medskip

\begin{verbatim}
  Node : list rose -> rose\end{verbatim}

\medskip\noindent
  then Coq generates the following induction rule\medskip

  \begin{verbatim}
  rose_ind : forall P : rose -> Prop,
               (forall l : list rose, P (Node l)) ->
               forall r : rose, P r\end{verbatim}
\medskip\noindent
But to prove a property of a rose tree \verb|Node l|, we must prove
that property assuming only that \verb|l| is a list of rose trees, and
without recourse to any induction hypothesis. There is, of course, a
presentation of rose trees by mutual recursion as well, but this
doesn't give the expected induction rule in Coq either. Intuitively,
what we expect is an induction rule whose premise is\medskip

  \begin{verbatim}
  forall [r_0, ..., r_n] : list rose,
       P(r_0) -> ... -> P(r_n) -> P(Node [r_0, ..., r_n])\end{verbatim}
\medskip\noindent
  The rule we derive for rose trees is indeed the expected one, which
  suggests that our derivation may enable automatic generation of more
  useful induction rules in Coq, rather than requiring the user to
  hand code them as is currently
  necessary. \smallskip
\item We further show in Example~\ref{ex:hereditary} how our generic
  induction rule can be instantiated, again to the families fibration
  over $\Set$, to derive a rule for the data type of finite hereditary
  sets. This data type is defined in terms of quotients and so lies
  outside
  most current theories of induction. \smallskip
\item Finally, we show in Example~\ref{ex:hyperfunctions} how our
  generic induction rule can be instantiated to the subobject
  fibration over $\omega CPO_{\bot}$ to derive a rule for the data
  type of hyperfunctions. Because this data type cannot be interpreted
  as a set, a fibration other than the families fibration over $\Set$
  is required; in this case, use of the subobject fibration allows us
  to derive an induction rule for admissible subsets of
  hyperfunctions. The ability to treat the data type of hyperfunctions
  thus underscores the importance of developing our results in the
  general fibrational framework.  Moreover, the functor underlying the
  data type of hyperfunctions is not strictly positive~\cite{gaa05},
  so the ability to treat this data type also underscores the
  advantage of being able to handle a very general class of functors
  going beyond simply polynomial functors. As far as we know,
  induction rules for finite hereditary sets and hyperfunctions have
  not previously existed in
  the general fibrational framework. \smallskip
\end{iteMize}

Although our theory of induction is applicable to {\em all} inductive
functors --- i.e., to {\em all} functors having initial algebras,
including those giving rise to nested types~\cite{mat09},
GADTs~\cite{she04}, indexed containers~\cite{am09}, dependent
types~\cite{nps90}, and inductive recursive types~\cite{dyb00} --- our
examples show that working in the general fibrational setting is
beneficial even if we restrict our attention to strictly positive data
types. We do, however, offer some preliminary thoughts in
Section~\ref{sec:conclusion} on the potentially delicate issue of
instantiating our general theory with fibrations appropriate for
deriving induction rules for specific classes of higher-order functors
of interest.  It is also worth noting that the specialisations of our
generic induction rule to polynomial functors in the families
fibration over $\Set$ coincide exactly with the induction rules of
Hermida and Jacobs.  But the structure we require of fibrations
generally is slightly different from that required by Hermida and
Jacobs, so while our theory is in essence a generalisation of theirs,
the two are, strictly speaking, incomparable. The structure we require
of our fibrations is, nevertheless, certainly present in all standard
fibrational models of type theory (see
Section~\ref{sec:fibrations}). Like Hermida and Jacobs, we prove our
generic induction rule correct by reducing induction to iteration. A
more detailed discussion of when our induction rules coincide with
those of Hermida and Jacobs is given in Section~\ref{sec:fibrations}.

We take a purely categorical approach to induction in this paper, and
derive our generic induction rule from only the initial algebra
semantics of data types. As a result, our work is inherently
extensional. Although translating our constructions into intensional
settings may therefore require additional effort, we expect the
guidance offered by the categorical viewpoint to support the
derivation of induction rules for functors that are not treatable at
present. Since we do not use any form of impredicativity in our
constructions, and instead use only the weaker assumption that initial
algebras exist, this guidance will be widely applicable.

The remainder of this paper is structured as follows. To make our
results as accessible as possible, we illustrate them in
Section~\ref{sec:stateofart} with a categorical derivation of the
familiar induction rule for the natural numbers. In
Section~\ref{sec:essence} we derive an induction rule for the special
case of the families fibration over $\Set$. We also show how this rule
can be instantiated to derive the one from
Section~\ref{sec:stateofart}, and the ones for rose trees and finite
hereditary sets mentioned above. Then, in Section~\ref{sec:fibrations}
we present our generic fibrational induction rule, establish a number
of results about it, and illustrate it with the aforementioned
application to hyperfunctions. The approach taken in this section is
completely different from the corresponding one in the conference
version of the paper~\cite{gjf10}, and allows us to improve upon and
extend our previous results. Section~\ref{sec:conclusion} concludes,
discusses possible instantiations of our generic induction rule for
higher-order functors, and offers some additional directions for
future research.

When convenient, we identify isomorphic objects of a category and
write $=$ rather than $\simeq$.  We write $1$ for the canonical
singleton set and denote its single element by $\cdot$\,. In
Sections~\ref{sec:stateofart} and~\ref{sec:essence} we assume that
types are interpreted as objects in $\Set$, so that $1$ also denotes
the unit type in those sections. We write $\mathit{id}$ for identity
morphisms in a category and $\mathit{Id}$ for the identity functor on
a category.

\section{A Familiar Induction Rule}\label{sec:stateofart}

Consider the inductive data type $\mathit{Nat}$, which defines the
natural numbers and can be specified in a programming language with
Haskell-like syntax by
\[\mathit{data \; Nat} \; = \; \mathit{Zero}\, \mid\, \mathit{Succ}\,
\mathit{Nat} \] 
\noindent
The observation that $\mathit{Nat}$ is the least fixed point of the
functor $N$ on $\Set$ --- i.e., on the category of sets and functions
--- defined by $N X \; = \; 1 \,+\, X$ can be used to define the
following iteration operator:
\[\begin{array}{lcl}  
\mathit{foldNat} & \; = \; & X \ra (X \ra X) \ra \mathit{Nat} \ra X\\
\mathit{foldNat}\,z\,s \,\mathit{Zero} & = & z\\
\mathit{foldNat}\,z\,s \,(\mathit{Succ}\,n) & = & s \,
(\mathit{foldNat} \,z \,s \, n)\\
\end{array}\]
The iteration operator $\mathit{foldNat}$ provides a uniform means of
expressing common and naturally occurring patterns of recursion over
the natural numbers.

Categorically, iteration operators such as $\mathit{foldNat}$ arise
from the initial algebra semantics of data types, in which every data
type is regarded as the carrier of the initial algebra of a functor
$F$.  If $\B$ is a category and $F$ is a functor on $\B$, then an {\em
  $F$-algebra} is a morphism $h : F X \rightarrow X$ for some object
$X$ of $\B$. We call $X$ the {\em carrier} of $h$. For any functor
$F$, the collection of $F$-algebras itself forms a category $\Alge{F}$
which we call the {\em category of $F$-algebras}. In $\Alge{F}$, an
{\em $F$-algebra morphism} between $F$-algebras $h : F X \ra X$ and $g
: F Y \ra Y$ is a map $f : X \rightarrow Y$ such that the following
diagram commutes:
\[ \xymatrix{ F A \;\; \ar[r]^{F f} \ar[d]_{h} & F B \ar[d]^{g}\\ 
A \;\; \ar[r]^{f} & B}
\] 
When it exists, the initial $F$-algebra $\mathit{in} : F (\mu F) \ra
\mu F$ is unique up to isomorphism and has the least fixed point $\mu
F$ of $F$ as its carrier. Initiality ensures that there is a unique
$F$-algebra morphism $\mathit{fold} \ h : \mu F \ra X$ from
$\mathit{in}$ to any $F$-algebra $h: F X \ra X$. This gives rise to
the following iteration operator $\mathit{fold}$ for the inductive
type $\mu F$:
\[\begin{array}{lcl}
\mathit{fold}  & : & (F X \ra X) \ra \mu F \ra X\\
\mathit{fold}\, h \, (\mathit{in}\, t) & \;\;= \;\;& h \,
(F\, (\mathit{fold}\, h)\, t) 
\end{array}\]
Since $\mathit{fold}$ is derived from initial algebra semantics it is
principled and correct. It is also expressive, since it
can be defined for {\em every} inductive type. In fact, $\mathit{fold}$
is a single iteration operator parameterised over inductive types rather
than a family of iteration operators, one for each such
type, and the iteration operator $\mathit{foldNat}$ above is the
instantiation to $\mathit{Nat}$ of the generic iteration operator
$\mathit{fold}$.  

The iteration operator $\mathit{foldNat}$ can be used to derive the
standard induction rule for $\mathit{Nat}$ which coincides with the
standard induction rule for natural numbers, i.e., with the familiar
principle of mathematical induction.  This rule says that if a
property $P$ holds for $0$, and if $P$ holds for $n + 1$ whenever it
holds for a natural number $n$, then $P$ holds for all natural
numbers. Representing each property of natural numbers as a predicate
$P \,: \,\mathit{Nat} \ra \Set$ mapping each $n : \mathit{Nat}$ to the
set of proofs that $P$ holds for $n$, we wish to represent this rule
at the object level as a function $\mathit{indNat}$ with type
\[\begin{array}{lcll}
& & \forall (P : \mathit{Nat} \ra \Set). &
P\, \mathit{Zero} \ra (\forall n : \mathit{Nat}.\, P\,n \ra P \,
(\mathit{Succ}\,n)) \ra (\forall n : \mathit{Nat}. \, P\, n)\\
\end{array}\]
Code fragments such as the above, which involve quantification over
sets, properties, or functors, are to be treated as ``categorically
inspired'' within this paper. This is because quantification over such
higher-kinded objects cannot be interpreted in $\Set$. In order to
give a formal interpretation to code fragments like the one above, we
would need to work in a category such as that of modest sets. While
the ability to work with functors over categories other than $\Set$ is
one of the motivations for working in the general fibrational setting
of Section~\ref{sec:fibrations}, formalising the semantics of such
code fragments would obscure the central message of this paper. Our
decision to treat such fragments as categorically inspired is
justified in part by the fact that the use of category theory to
suggest computational constructions has long been regarded as fruitful
within the functional programming community (see,
e.g.,~\cite{bdm97,bm98,mog91}).

A function $\mathit{indNat}$ with the above type takes as input the
property $P$ to be proved, a proof $\phi$ that $P$ holds for
$\mathit{Zero}$, and a function $\psi$ mapping each $n : \mathit{Nat}$
and each proof that $P$ holds for $n$ to a proof that $P$ holds for
$\mathit{Succ}\,n$, and returns a function mapping each $n :
\mathit{Nat}$ to a proof that $P$ holds for $n$, i.e., to an element
of $P\,n$.  We can write $\mathit{indNat}$ in terms of
$\mathit{foldNat}$ --- and thus reduce induction for $\mathit{Nat}$ to
iteration for $\mathit{Nat}$ --- as follows. First note that
$\mathit{indNat}$ cannot be obtained by instantiating the type $X$ in
the type of $\mathit{foldNat}$ to a type of the form $P n$ for a
specific $n$ because $\mathit{indNat}$ returns elements of the types
$P\, n$ for different values $n$ and these types are, in general,
distinct from one another. We therefore need a type containing all of
the elements of $P \,n$ for every $n$. Such a type can informally be
thought of as the union over $n$ of $P n$, and is formally given by
the dependent type $\Sigma n : \mathit{Nat}.\, P \,n$ comprising pairs
$(n,p)$ where $n : \mathit{Nat}$ and $p : P\, n$.

The standard approach to defining $\mathit{indNat}$ is thus to apply
$\mathit{foldNat}$ to an $N$-algebra with carrier $\Sigma n :
\mathit{Nat}. \,P\, n$. Such an algebra has components $\alpha :
\Sigma n : Nat.\, P\, n$ and $\beta : \Sigma n : Nat. \,P\, n \;\ra\;
\Sigma n : Nat.\, P\, n$.  Given $\phi : P\,\mathit{Zero}$ and $\psi :
\forall n. \,P\,n \ra P \,(\mathit{Succ}\,n)$, we choose $\alpha \; =
\; (\mathit{Zero}, \phi)$ and $\beta\, (n,p) \; = \;
(\mathit{Succ}\,n, \psi \, n \, p)$ and note that $foldNat \, \alpha
\, \beta : Nat \ra \Sigma n : Nat. \, P \,n$.  We tentatively take
$\mathit{indNat} \,\, P \,\, \phi \,\, \psi \,\, n$ to be $p$, where
$foldNat\, \alpha \, \beta \, n \; = \; (m, p)$. But in order to know
that $p$ actually gives a proof for $n$ itself, we must show that $m =
n$. Fortunately, this follows easily from the uniqueness of
$\mathit{foldNat} \, \alpha \, \beta$. Indeed, we have that
\[\xymatrix{ 1+Nat \;\; \ar[r]^{} \ar[d]_{in} & 1+\Sigma n :
  Nat.\, P\,n \ar[r]^{} \ar[d]_{[\alpha,\beta]} & 1+Nat \ar[d]_{in}\\
  Nat \ar[r]^{\!\!\!\!\!\!\!\!\!\!\!\!\!\!\!\!\!\mathit{foldNat}
    \,\alpha\,\beta} & \Sigma n : Nat.\,P\,n
  \ar[r]^{\;\;\;\;\;\;\;\;\lambda (n,p). \, n} & \mathit{Nat}}\]
\noindent
commutes and, by initiality of $\mathit{in}$, that $(\lambda (n,p). \,
n) \circ (foldNat \, \alpha\,\beta)$ is the identity map. Thus
\[
n = (\lambda (n,p). \, n) (foldNat \, \alpha\,\beta \, n) = (\lambda
(n,p). \, n) (m,p) = m
\]
Letting $\pi'_P$ be the second projection on dependent pairs involving
the predicate $P$, the induction rule for $\mathit{Nat}$ is thus
\[\begin{array}{lcl}
 \mathit{indNat} & \; : \;& \forall (P : \mathit{Nat} \ra \Set). \;
 P\, \mathit{Zero} \ra (\forall n : \mathit{Nat}.\, P\,n \ra P \,
 (\mathit{Succ}\,n))\\
 &   & \;\;\;\; \ra (\forall n : \mathit{Nat}. \, P\, n)\\
 \mathit{indNat}\; P\; \phi\; \; \psi & = & \pi'_P\; \circ 
 (\mathit{foldNat}\,\; (\mathit{Zero}, \phi) \,\;
 (\lambda (n,p).\, (\mathit{Succ}\,n, \psi \, n\,p)))\\  
\end{array}\]
\noindent
As expected, this induction rule states that, for every property $P$,
to construct a proof that $P$ holds for every $n : \mathit{Nat}$, it
suffices to provide a proof that $P$ holds for $\mathit{Zero}$, and to
show that, for any $n : \mathit{Nat}$, if there is a proof that $P$
holds for $n$, then there is also a proof that $P$ holds for
$\mathit{Succ} \,n$.

The use of dependent types is fundamental to this formalization of the
induction rule for $\mathit{Nat}$, but this is only possible because
properties to be proved are taken to be set-valued functions. The
remainder of this paper uses fibrations to generalise the above
treatment of induction to arbitrary inductive functors and arbitrary
properties which are suitably fibred over the category whose objects
interpret types. In the general fibrational setting, properties are
given axiomatically via the fibrational structure rather than assumed
to be (set-valued) functions.

\section{Induction Rules for Predicates over $\Set$}
\label{sec:essence} 

The main result of this paper is the derivation of a sound induction
rule that is generic over all inductive types and which can be used to
verify any notion of property that is fibred over the category whose
objects interpret types. In this section we assume that types are
modelled by sets, so the functors we consider are on $\Set$ and the
properties we consider are functions mapping data to sets of proofs
that these properties hold for them. We make these assumptions because
it allows us to present our derivation in the simplest setting
possible, and also because type theorists often model properties in
exactly this way. This makes the present section more accessible and,
since the general fibrational treatment of induction can be seen as a
direct generalisation of the treatment presented here,
Section~\ref{sec:fibrations} should also be more easily digestible
once the derivation is understood in this special case. Although the
derivation of this section can indeed be seen as the specialisation of
that of Section~\ref{sec:fibrations} to the families fibration over
$\Set$, no knowledge of fibrations is required to understand it
because all constructions are given concretely rather than in their
fibrational forms.

We begin by considering what we might naively expect an induction rule
for an inductive data type $\mu F$ to look like. The derivation for
$\mathit{Nat}$ in Section~\ref{sec:stateofart} suggests that, in
general, it should look something like this:
\[\mathit{ind} \; : \; \forall P : \mu \,F \ra \Set.\;\; 
       ??? \; \ra \; \forall \, x : \mu F.\,  P\, x
\]
\noindent
But what should the premises --- denoted $???$ here --- of the generic
induction rule $\mathit{ind}$ be? Since we want to construct, for any
term $x : \mu F$, a proof term of type $P\,x$ from proof terms for
$x$'s substructures, and since the functionality of the
$\mathit{fold}$ operator for $\mu F$ is precisely to compute a value
for $x : \mu F$ from the values for $x$'s substructures, it is natural
to try to equip $P$ with an $F$-algebra structure that can be input to
$\mathit{fold}$ to yield a mapping of each $x : \mu F$ to an element
of $P\,x$. But this approach quickly hits a snag. Since the codomain
of every predicate $P : \mu F \ra \Set$ is $\Set$ itself, rather than
an object of $\Set$, $F$ cannot be applied to $P$ as is needed to
equip $P$ with an $F$-algebra structure. Moreover, an induction rule
for $\mu F$ cannot be obtained by applying $\mathit{fold}$ to an
$F$-algebra with carrier $P\,x$ for any specific $x$. This suggests
that we should try to construct an $F$-algebra not for $P\,x$ for each
term $x$, but rather for $P$ itself.

Such considerations led Hermida and Jacobs~\cite{hj98} to define a
category of predicates $\mathcal{P}$ and a lifting for each polynomial
functor $F$ on $\Set$ to a functor $\hat{F}$ on ${\mathcal{P}}$ that
respects the structure of $F$. They then constructed
$\hat{F}$-algebras with carrier $\mathcal{P}$ to serve as the premises
of their induction rules. The crucial part of their construction,
namely the lifting of polynomial functors, proceeds inductively and
includes clauses such as
\[(\widehat{F+G}) \;P = \hat{F} P + \hat{G} P\]
\noindent
and
\[(\widehat{F \times G}) \; P = \hat{F} P \times \hat{G} P\] 

\vspace*{0.05in}

\noindent
The construction of Hermida and Jacobs is very general: they consider
functors on bicartesian categories rather than just on $\Set$, and
represent properties by bicartesian fibrations over such categories
instead of using the specific notion of predicate from
Definition~\ref{def:preds} below. On the other hand, they define
liftings for polynomial functors.

The construction we give in this section is in some sense orthogonal
to Hermida and Jacobs': we focus exclusively on functors on $\Set$ and
a particular category of predicates, and show how to define liftings
for all inductive functors on $\Set$, including non-polynomial
ones. In this setting, the induction rule we derive properly extends
Hermida and Jacobs', thus catering for a variety of data types that
they cannot treat. In the next section we derive analogous results in
the general fibrational setting. This allows us to derive sound
induction rules for initial algebras of functors defined on categories
other than $\Set$ which can be used to prove arbitrary properties that
are suitably fibred over the category interpreting types.

We begin with the definition of a predicate.

\begin{defi}\label{defn:predicates}
  Let $X$ be a set. A {\em predicate on $X$} is a function $P : X \ra
  \Set$ mapping each $x \in X$ to a set $P\,x$. We call $X$ the {\em
    domain} of $P$.
\end{defi}

We may speak simply of ``a predicate $P$'' if the domain of $P$ is
understood. A predicate $P$ on $X$ can be thought of as mapping each
element $x$ of $X$ to the set of proofs that $P$ holds for $x$.  We
now define our category of predicates.

\begin{defi}\label{def:preds}
  The {\em category of predicates} ${\mathcal{P}}$ has predicates as
  its objects. A morphism from a predicate $P : X \ra \Set$ to a
  predicate $P' : X' \ra \Set$ is a pair $(f,f^\sim):P \rightarrow P'$
  of functions, where $f:X \rightarrow X'$ and $f^\sim:\forall x :
  X.\, P \,x \rightarrow P' (f \,x)$. Composition of predicate
  morphisms is given by $(g,g^\sim) \circ (f,f^\sim) = (g \circ f,\,
  \lambda x p. \,g^\sim (fx)(f^\sim xp))$.
\end{defi}
\noindent
Diagrammatically, we have
\[ \xymatrix{X \;\; \ar[rr]^{f} \ar[dr]_P^{}="a" & & X' \ar[dl]^{P'}_{}="b" \\
 & \Set  \ar"a";"b"^{f^\sim\ } &}
\]
\noindent
As the diagram indicates, the notion of a morphism from $P$ to $P'$
does not require the sets of proofs $P \,x$ and $P'\,(f\, x)$, for any
$x \in X$, to be {\em equal}. Instead, it requires only the existence
of a function $f^\sim$ which maps, for each $x$, each proof in $P \,x$
to a proof in $P'\,(f\, x)$. We denote by $U: {\mathcal{P}}
\rightarrow \Set$ the {\em forgetful functor} mapping each predicate
$P : X \ra \Set$ to its domain $X$ and each predicate morphism $(f,
f^\sim)$ to $f$.

An alternative to Definition~\ref{def:preds} would take the category
of predicates to be the arrow category over $\Set$, but the natural
lifting in this setting does not indicate how to generalise liftings
to other fibrations. Indeed, if properties are modelled as functions,
then every functor can be applied to a property, and hence every
functor can be its own lifting. In the general fibrational setting,
however, properties are not necessarily modelled by functions, so a
functor cannot, in general, be its own lifting. The decision not to
use arrow categories to model properties is thus dictated by our
desire to lift functors in a way that indicates how liftings can be
constructed in the general fibrational setting.

We can now give a precise definition of a lifting.

\begin{defi}\label{def:liftingA}
  Let $F$ be a functor on $\Set$. A {\em lifting} of $F$ from $\Set$
  to ${\mathcal{P}}$ is a functor $\hat{F}$ on ${\mathcal{P}}$ such
  that the following diagram commutes:
\[ \xymatrix{ {\mathcal{P}} \;\; \ar[r]^{\hat{F}} \ar[d]_{U} &
  {\mathcal{P}} \ar[d]^U \\ 
\Set \ar[r]_{F} & \Set}
\]
\end{defi}
\noindent
We can decode the definition of $\hat{F}$ as follows. The object part
of $\hat{F}$ must map each predicate $P : X \rightarrow \Set$ to a
predicate $\hat{F}P: F\, X \rightarrow \Set$, and thus can be thought
of type-theoretically as a function $\forall (X:\Set). \; (X
\rightarrow \Set) \rightarrow F\, X \rightarrow \Set$. Of course,
$\hat{F}$ must also act on morphisms in a functorial manner.

We can now use the definition of a lifting to derive the standard
induction rule from Section~\ref{sec:stateofart} for $\mathit{Nat}$ as
follows.

\begin{ex}\label{ex:nats}
  The data type of natural numbers is $\mu N$ where $N$ is the functor
  on $\Set$ defined by $N\, X = 1 + X$. A lifting $\hat{N}$ of $N$ can
  be defined by sending each predicate $P : X \ra \Set$ to the
  predicate $\widehat{N}P : NX \ra \Set$ given by
\[\begin{array}{lll}
\hat{N} P\, (\mathit{inl} \,\cdot) & \;=\; & 1 \\
\hat{N} P \,(\mathit{inr} \,n) & \;=\; & P \,n  
\end{array}\]
An $\hat{N}$-algebra with carrier $P:\Nat \rightarrow \Set$ can
be given by $\mathit{in} : 1 + \Nat \rightarrow \Nat$ and $in^\sim
\; : \; \forall t : 1 + \Nat. \;\hat{N} P \,t \ra P (\mathit{in}
\,t)$. Since $\mathit{in} \,(\mathit{inl} \, \cdot) = 0$ and
$\mathit{in} \,(\mathit{inr} \, n) = n + 1$, we see that $in^\sim$
consists of an element $h_1 : P \,0$ and a function $h_2 : \forall n :
\Nat. \; P\, n \ra P\,(n+1)$.  Thus, the second component
$\mathit{in}^\sim$ of an $\hat{N}$-algebra with carrier $P :
\mathit{Nat} \ra \Set$ and first component $\mathit{in}$ gives the
premises of the familiar induction rule in Example~\ref{ex:nats}.
\end{ex}

The notion of predicate comprehension is a key ingredient of our
lifting. It begins to explain, abstractly, what the use of
$\Sigma$-types is in the theory of induction, and is the key construct
allowing us to define liftings for non-polynomial, as well as
polynomial, functors.

\begin{defi}\label{def:comprehension}
  Let $P$ be a predicate on $X$. The {\em comprehension} of $P$,
  denoted $\{P\}$, is the type $\Sigma x:X.\, P\,x$ comprising pairs
  $(x, p)$ where $x : X$ and $p : P x$.  The map taking each predicate
  $P$ to $\{P\}$, and taking each predicate morphism $(f,f^\sim) : P
  \ra P'$ to the morphism $\{(f,f^\sim)\} : \{P\} \ra \{P'\}$ defined
  by $\{(f,f^\sim)\} (x,p) = (f x, f^\sim x \,p)$, defines the {\em
    comprehension functor} $\{-\}$ from ${\mathcal{P}}$ to $\Set$.
\end{defi}

We are now in a position to define liftings uniformly for all
functors:

\begin{defi}\label{def:liftingC}
  If $F$ is a functor on $\Set$, then the lifting $\hat{F}$ is the
  functor on ${\mathcal{P}}$ given as follows. For every predicate $P$
  on $X$, $\hat{F}\, P \,:\, F\,X \ra \Set$ is defined by $\hat{F}\,P
  = (F \, \pi_{P})^{-1}$, where the natural transformation $\pi :
  \{-\} \ra U$ is given by $\pi_{P} \, (x,p) = x$. For every predicate
  morphism $f : P \ra P'$, $\hat{F} f = (k,k^\sim)$ where $k = F U f$,
  and $k^\sim : \forall y : FX. \, \hat{F}P \,y \ra \hat{F}P'\, (k \,
  y)$ is given by $k^\sim \, y \, z = F\{f\} z$.
\end{defi}
\noindent
In the above definition, note that the inverse image $f^{-1}$ of $f :
X \rightarrow Y$ is indeed a predicate $P : Y \ra \Set$. Thus if $P$
is a predicate on $X$, then $\pi_P : \{P\} \ra X$ and $F \pi_P : F
\{P\} \ra FX$.  Thus $\hat{F}P$ is a predicate on $FX$, so $\hat{F}$
is a lifting of $F$ from $\Set$ to ${\mathcal{P}}$. The lifting
$\hat{F}$ captures an ``all'' modality, in that it generalises
Haskell's \verb|all| function on lists to arbitrary data types. A
similar modality is given in~\cite{mor07} for indexed containers.

The lifting in Example~\ref{ex:nats} is the instantiation of the
construction in Definition~\ref{def:liftingC} to the functor $N X = 1
+ X$ on $\Set$. Indeed, if $P$ is any predicate, then $\hat{N}\,P =
(N\, \pi_P)^{-1}$, i.e., $\hat{N}\,P = (\mathit{id} +
\pi_P)^{-1}$. Then, since the inverse image of the coproduct of
functions is the coproduct of their inverse images, since
$\mathit{id}^{-1}\, 1 = 1$, and since $\pi_P^{-1} n = \{(n,p) \mid p :
P n\}$ for all $n$, we have $\hat{N}\,P \, (inl \, \cdot) = 1$ and
$\hat{N}\,P \, (inr \, n) = P\,n$. As we will see, a similar situation
to that for $\mathit{Nat}$ holds in general: for any functor $F$ on
$\Set$, the second component of an $\hat{F}$-algebra whose carrier is
the predicate $P$ on the data type $\mu F$ and whose first component
is $\mathit{in}$ gives the premises of an induction rule that can be
used to show that $P$ holds for all data of type $\mu F$.

The rest of this section shows that $F$-algebras with carrier $\{P\}$
are interderivable with $\hat{F}$-algebras with carrier $P$, and then
uses this result to derive our induction rule.

\begin{defi}\label{def:truth}
  The functor $K_1:\Set \rightarrow {\mathcal{P}}$ maps each set $X$
  to the predicate $K_1 X = \lambda x:X.\, 1$ on $X$ and each $f : X
  \ra Y$ to the predicate morphism $(f, \lambda x:X.\,id)$.
\end{defi}

\noindent
The predicate $K_1X$ is called the {\em truth predicate on $X$}. For
every $x : X$, the set $K_1 X x$ of proofs that $K_1 X$ holds for $x$
is a singleton, and thus is non-empty. We intuitively think of a
predicate $P : X \ra \Set$ as being true if $P x$ is non-empty for
every $x : X$. We therefore consider $P$ to be true if there exists a
predicate morphism from $K_1X$ to $P$ whose first component is
$\mathit{id}_X$.  For any functor $F$, the lifting $\hat{F}$ is {\em
  truth-preserving}, i.e., $\hat{F}$ maps the truth predicate on any
set $X$ to that on $F X$.

\begin{lem}\label{lem:aux}
 For any functor $F$ on $\Set$ and any set $X$, $\hat{F}(K_1X) \cong
 K_1(FX)$.
\end{lem}

\begin{proof}
By Definition~\ref{def:liftingC}, $\hat{F}(K_1
X) = (F\pi_{K_1 X})^{-1}$. We have that $\pi_{K_1 X}$ is an
isomorphism because there is only one proof of $K_1 X$ for each $x :
X$, and thus that $F\,\pi_{K_1 X}$ is an isomorphism as well. As a
result, $(F\,\pi_{K_1 X})^{-1}$ maps every $y : F X$ to a singleton
set, and therefore $\hat{F}(K_1 X) = (F\pi_{K_1 X})^{-1} \cong \lambda y :
F X.\, 1 = K_1 (F X)$.
\end{proof}

\vspace{0.1in}

\noindent 

The fact that $K_1$ is a left-adjoint to $\{-\}$ is critical to the
constructions below. This is proved in~\cite{hj98}; we include its
proof here for completeness and to establish notation. The description
of comprehension as a right adjoint can be traced back to
Lawvere~\cite{law70}.

\begin{lem}\label{lem:adjoints}
$K_1$ is left adjoint to $\{-\}$.
\end{lem}

\begin{proof}
  We must show that, for any predicate $P$ and any set $Y$, the set
  ${{\mathcal{P}}}(K_1Y,P)$ of morphisms from $K_1Y$ to $P$ in
  ${\mathcal{P}}$ is in bijective correspondence with the set $\Set
  (Y,\{P\})$ of morphisms from $Y$ to $\{P\}$ in $\Set$.  Define maps
  $(-)^\dagger : \Set(Y, \{P\}) \ra {{\mathcal{P}}}(K_1 Y, P)$ and
  $(-)^\hash : {{\mathcal{P}}}(K_1 Y, P) \ra \Set(Y, \{P\})$ by
  $h^\dagger = (h_1, h_2)$ where $h y = (v, p), \,h_1 y = v$ and $h_2
  y = p$, and $(k, k^\sim)^\hash = \lambda (y : Y).\,(k y, k^\sim y)$.
  These give a natural isomorphism between $\Set(Y, \{P\})$ and
  ${{\mathcal{P}}}(K_1 Y, P)$. 
\end{proof}

\vspace*{0.1in}

Naturality of $(-)^\dagger$ ensures that $(g \circ f)^\dagger \; = \;
g^\dagger \circ K_1 f$ for all $f : Y' \ra Y$ and $g : Y \ra
\{P\}$. Similarly for $(-)^\hash$. Moreover, $id^\dagger$ is
the counit, at $P$, of the adjunction between $K_1$ and $\{-\}$. These
observations are used in the proof of
Lemma~\ref{lem:Psi}. Lemmas~\ref{lem:Phi} and~\ref{lem:Psi} are the
key results relating $F$-algebras and $\hat{F}$ algebras, i.e.,
relating iteration and induction. They are special cases of Theorem
~\ref{thm:adj} below, but we include their proofs to ensure continuity
of our presentation and to ensure that this section is self-contained.

We first we show how to construct $\hat{F}$-algebras from $F$-algebras

\begin{lem}\label{lem:Phi}
 There is a functor $\Phi:\Alge{F} \rightarrow \Alge{\hat{F}}$ such
 that if $k:FX \rightarrow X$, then $\Phi k : \hat{F}(K_1 X)
 \rightarrow K_1 X$.
\end{lem}

\begin{proof} 
  For an $F$-algebra $k:FX \rightarrow X$ define $\Phi k = K_1 k$, and
  for two $F$-algebras $k:FX \rightarrow X$ and $k':FX' \rightarrow
  X'$ and an $F$-algebra morphism $h : X \ra X'$ between them define
  the $\hat{F}$-algebra morphism $\Phi h : \Phi k \rightarrow \Phi k'$
  by $\Phi h = K_1 h$. Then $K_1(FX) \cong \hat{F}(K_1X)$ by
  Lemma~\ref{lem:aux}, so that $\Phi k$ is an $\hat{F}$-algebra and
  $K_1 h$ is an $\hat{F}$-algebra morphism. It is easy to see that
  $\Phi$ preserves identities and composition.
\end{proof}

\noindent
We can also construct $F$-algebras from $\hat{F}$-algebras.

\begin{lem}\label{lem:Psi}
The functor $\Phi$ has a right adjoint $\Psi$ such that if $j :
\hat{F}P \ra P$, then $\Psi j: F \{P\} \ra \{P\}$. 
\end{lem}

\begin{proof} 
  We construct the adjoint functor $\Psi:\Alge{\hat{F}} \rightarrow
  \Alge{F}$ as follows. Given an $\hat{F}$-algebra $j:\hat{F} P
  \rightarrow P$, we use the fact that $\hat{F}(K_1\{P\}) \cong
  K_1(F\{P\})$ by Lemma~\ref{lem:aux} to define $\Psi j : F \{P\}
  \rightarrow \{P\}$ by $\Psi j = (j \circ \hat{F}
  \,id^\dagger)^\hash$.  To specify the action of $\Psi$ on
  an $\hat{F}$-algebra morphism $h$, define $\Psi h = \{h\}$.  Clearly
  $\Psi$ preserves identity and composition.

  Next we show $\Phi \dashv \Psi$, i.e., for every $F$-algebra $k:FX
  \rightarrow X$ and $\hat{F}$-algebra $j:\hat{F} P \rightarrow P$
  with $P$ a predicate on $X$, there is a natural isomorphism between
  $F$-algebra morphisms from $k$ to $\Psi j$ and $\hat{F}$-algebra
  morphisms from $\Phi k$ to $j$. We first observe that an $F$-algebra
  morphism from $k$ to $\Psi j$ is a map from $X$ to $\{P\}$, and an
  $\hat{F}$-algebra morphism from $\Phi k$ to $j$ is a map from $K_1
  X$ to $P$. A natural isomorphism between such maps is given by the
  adjunction $K_1 \dashv \{-\}$ from Lemma~\ref{lem:adjoints}.  We
  must check that $f:X \rightarrow \{P\}$ is an $F$-algebra morphism
  from $k$ to $\Psi j$ iff $f^\dagger : K_1 X \rightarrow P$ is an
  $\hat{F}$-algebra morphism from $\Phi k$ to $j$.

  To this end, assume $f:X \rightarrow \{P\}$ is an $F$-algebra
  morphism from $k$ to $\Psi j$, i.e., assume $f \circ k = \Psi j
  \circ F f$. We must prove that $f^\dagger \circ\, \Phi k = j \circ
  \hat{F} f^\dagger$. By the definition of $\Phi$ in
  Lemma~\ref{lem:Phi}, this amounts to showing $f^\dagger \circ K_1 k
  = j \circ \hat{F} f^\dagger$. Now, since $(-)^\dagger$ is an
  isomorphism, $f$ is an $F$-algebra morphism iff $(f \circ k)^\dagger
  = (\Psi j \circ F f)^\dagger$. Naturality of $(-)^\dagger$ ensures
  that $(f \circ k)^\dagger = f^\dagger \circ K_1 k$ and that $(\Psi j
  \circ F f)^\dagger = (\Psi j)^\dagger\circ K_1 (F f)$, so the
  previous equality holds iff
  \begin{eqnarray}\label{eqn:aux}
    f^\dagger \circ K_1 k & = & (\Psi j)^\dagger\circ K_1 (F f)
  \end{eqnarray}
  But
  \[\begin{array}{lll}
    & j \circ \hat{F} f^\dagger & \\
    = & j \circ \hat{F} \mathit{id}^\dagger \circ K_1f) & \mbox{ by
      naturality of } (-)^\dagger \mbox{ and } f = id \circ f\\
    = &(j \circ \hat{F}\, id^\dagger) \circ \hat{F}(K_1f) &
    \mbox{ by the functoriality of } \hat{F}\\ 
    = & (\Psi j)^\dagger \circ K_1(F f) & \mbox{ by the definition of } \Psi,
    \mbox{ the fact that } (-)^\dagger \mbox{ and } (-)^\hash\\ 
    & & \mbox{\;\;\; are inverses, and Lemma~\ref{lem:aux}}\\   
    = & f^\dagger \circ K_1 k & \mbox{ by Equation~\ref{eqn:aux}}\\
  \end{array}\]
  Thus, $f^\dagger$ is indeed an $\hat{F}$-algebra morphism from $\Phi
  k$ to $j$. 
\end{proof}

\vspace*{0.1in}

Lemma~\ref{lem:Psi} ensures that $F$-algebras with carrier $\{P\}$ are
interderivable with $\hat{F}$-algebras with carrier $P$.  For example,
the $N$-algebra $[\alpha,\beta]$ with carrier $\{P\}$ from
Section~\ref{sec:stateofart} can be derived from the $\hat{N}$-algebra
with carrier $P$ given in Example~\ref{ex:nats}.  Since we define a
lifting $\hat{F}$ for any functor $F$, Lemma~\ref{lem:Psi} thus shows
how to construct $F$-algebras with carrier $\Sigma x : \mu F.\, P x$
for any functor $F$ and predicate $P$ on $\mu F$.

\begin{cor}\label{cor:fixed-point}
  For any functor $F$ on $\Set$, the predicate $K_1 (\mu F)$ is the
  carrier of the initial $\hat{F}$-algebra.
\end{cor}

\begin{proof} 
  Since $\Phi$ is a left adjoint it preserves initial objects, so
  applying $\Phi$ to the initial $F$-algebra $in :F(\mu F) \rightarrow
  \mu F$ gives the initial $\hat{F}$-algebra. By Lemma~\ref{lem:Phi},
  $\Phi\,in$ has type $\hat{F}(K_1 (\mu F)) \ra K_1 (\mu F)$, so the
  carrier of the initial $\hat{F}$-algebra is $K_1(\mu F)$.
\end{proof}

\vspace*{0.1in}

We can now derive our generic induction rule. For every predicate $P$
on $X$ and every $\hat{F}$-algebra $(k,k^\sim) : \hat{F} P \rightarrow
P$, Lemma~\ref{lem:Psi} ensures that $\Psi$ constructs from
$(k,k^\sim)$ an $F$-algebra with carrier $\{P\}$. 
Applying the iteration operator to this algebra gives a map
\[
\mathit{fold} \;(\Psi \,(k,k^\sim)) : \mu F \rightarrow \{P\}
\]
This map decomposes into two parts: $\phi\, = \, \pi_P \circ
\mathit{fold} \;(\Psi \,(k,k^\sim)) : \mu F \rightarrow X$ and $\psi\,
: \,\forall (t : \mu F). \,P\,(\phi\, t)$.  Initiality of $in : F(\mu
F) \ra \mu F$, the definition of $\Psi$, and the naturality of $\pi_P$
ensure $\phi = \mathit{fold}\,k$.  Recalling that $\pi'_P$ is the
second projection on dependent pairs involving the predicate $P$, this
gives the following sound generic induction rule for the type $X$,
which reduces induction to iteration:
\[\begin{array}{lll}
\mathit{genind} & : & \forall \, (F : \Set \rightarrow \Set) \; (P
: X \ra \Set) \; ((k,k^\sim) : (\hat{F}P \ra P)) \;(t : \mu F).\\
 & & \;\;\;\;\;\; P\, (\mathit{fold} \, k \, t)\\ 
\mathit{genind} \, F \, P & = & \pi_P' \circ \mathit{fold} \circ \Psi
\end{array}\]
\noindent
Notice this induction rule is actually capable of dealing with
predicates over arbitrary sets and not just predicates over $\mu
F$. However, when $X = \mu F$ and $k = in$, initiality of $in$ further
ensures that $\phi \,=\, \mathit{fold} \, in \, = \, id$, and thus
that $\mathit{genind}$ specialises to the expected induction rule for
an inductive data type $\mu F$:
\[\begin{array}{lll}
\mathit{ind} & : & \forall \, (F : \Set \rightarrow \Set) \;
(P : \mu F \ra \Set) \;
((k,k^\sim) : (\hat{F}P \ra P)). \\
& & \;\;\;\;\;\;(k = in) \ra \forall (t : \mu F).\,P\, t\\ 
\mathit{ind} \, F \, P & = & \pi'_P \circ \mathit{fold} \circ \Psi
\end{array}\]
This rule can be instantiated to familiar rules for polynomial data
types, as well as to ones we would expect for data types such as rose
trees and finite hereditary sets, both of which lie outside the scope
of Hermida and Jacobs' method. 

\begin{ex}\label{ex:rose}
The data type of rose trees is given in Haskell-like syntax by
\[\mathit{data \; Rose} = \mathit{Node} (\mathit{List
 \,Rose})\]
\noindent
The functor underlying $Rose$ is $F X = \mathit{List} \, X$ and its
induction rule is 
\[\begin{array}{lll}
\mathit{indRose} & \;:\; & \forall \; (P
: \mathit{Rose} \ra \Set) \; ((k,k^\sim) : (\hat{F}P \ra P)).\\
               &       & \;\;\;\; (k = \mathit{in}) \ra \forall (x :
               \mathit{Rose}).\,P\, x\\  
\mathit{indRose} \, F \, P & \; = \; & \pi_P' \circ \mathit{fold} \circ \Psi
\end{array}\]
\noindent
Calculating $\hat{F}P = (F\pi_P)^{-1} : F \, \mathit{Rose} \ra \Set$,
and writing $xs \, !! \, k$ for the $k^{th}$ component of a list $xs$,
we have that
\[\begin{array}{lll}
 &   & \hat{F}\, P\, rs\\
 & = & \{z : F \{P\} \mid F \, \pi_p \, z = rs\}\\
 & = & \{cps : \mathit{List} \,
 \{P\} \mid \mathit{List} \;
 \pi_P \; cps = rs\}\\
 & = & \{cps : \mathit{List} \,
 \{P\} \mid \forall k < \mathit{length} \, 
 cps. \, \pi_P \, (cps \, !!\, k) = rs\, !!\, k\}
\end{array}\]
An $\hat{F}$-algebra whose underlying $F$-algebra is $\mathit{in}
: F  \, \mathit{Rose} \ra \mathit{Rose}$ is thus a pair of functions
$(\mathit{in}, k^\sim)$, where $k^\sim$ has type
\[\begin{array}{lll}
 & = & \forall rs : \mathit{List \,
   Rose}.\\
 &   &\;\;\; \{cps : \mathit{List} \, \{P\} \mid \forall k <
 \mathit{length} \, cps. \, \pi_P \, (cps \, !!\, k) = rs\, !!\, k\}
 \; \ra \; P \, (\mathit{Node} \, rs)\\
 & = & \forall rs : \mathit{List \,
   Rose}.\; (\forall k < \mathit{length} \; rs. \, P \, (rs \,!! \,k))
 \ra P (\mathit{Node} \, rs)\\
\end{array}\]
\noindent
The last equality is due to surjective pairing for dependent products
and the fact that $\mathit{length} \, cps = \mathit{length} \,
rs$. The type of $k^\sim$ gives the hypotheses of the induction rule
for rose trees.

\end{ex}

Although finite hereditary sets are defined in terms of quotients, and
thus lie outside the scope of previously known methods, they can be
treated with ours.

\begin{ex}\label{ex:hereditary}
  Hereditary sets are sets whose elements are themselves sets, and so
  are the core data structures within set theory. The data type $HS$
  of finitary hereditary sets is $\mu P_f$ for the finite powerset
  functor $P_f$.  We can derive an induction rule for finite
  hereditary sets as follows.  If $P:X \rightarrow \Set$, then $P_f
  \pi_P : P_f (\Sigma x:X. P x) \rightarrow P_f X$ maps each set
  $\{(x_1,p_1), \ldots, (x_n,p_n)\}$ to the set $\{x_1, \ldots,
  x_n\}$, so that $(P_f \pi_P)^{-1}$ maps a set $\{x_1, \ldots, x_n\}$
  to the set $P x_1 \times \ldots \times P x_n$.  A
  $\hat{P_f}$-algebra with carrier $P : HS \rightarrow \Set$ and first
  component $in$ therefore has as its second component a function of
  type
\[
\forall (\{s_1, \ldots, s_n\} : P_f(HS)). \, P s_1 \times \ldots \times
P s_n \rightarrow P(in \{s_1, \ldots, s_n\})
\]
The induction rule for finite hereditary sets is thus
\[\begin{array}{l}
indHS :: (\forall (\{s_1, \ldots, s_n\} : P_f(HS)). \, P s_1 \times
\ldots \times P s_n \rightarrow P(in \{s_1, \ldots, s_n\}))\\
\;\;\;\;\;\;\;\;\;\;\;\;\;\;\;\;\;\;\;\;\rightarrow \forall (s : HS). P\,s
\end{array}\]
\end{ex}

\section{Generic Fibrational Induction Rules}\label{sec:fibrations}

We can treat more general notions of predicates using fibrations. We
motivate the use of fibrations by observing that i) the semantics of
data types in languages involving recursion and other effects usually
involves categories other than $\Set$; ii) in such circumstances, the
notion of a predicate can no longer be taken as a function with
codomain $\Set$; and iii) even when working in $\Set$ there are
reasonable notions of ``predicate'' other than that in
Section~\ref{sec:essence}. (For example, a predicate on a set $X$
could be a subobject of $X$).  Moreover, when, in future work, we
consider induction rules for more sophisticated classes of data types
such as indexed containers, inductive families, and inductive
recursive families (see Section~\ref{sec:conclusion}), we will not
want to have to develop an individual {\em ad hoc} theory of induction
for each such class. Instead, we will want to appropriately
instantiate a single generic theory of induction. That is, we will
want a uniform axiomatic approach to induction that is widely
applicable, and that abstracts over the specific choices of category,
functor, and predicate giving rise to different forms of induction for
specific classes of data types.

Fibrations support precisely such an axiomatic approach. This section
therefore generalises the constructions of the previous one to the
general fibrational setting. The standard model of type theory based
on locally cartesian closed categories does arise as a specific
fibration --- namely, the codomain fibration over $\Set$ --- and this
fibration is equivalent to the families fibration over $\Set$. But the
general fibrational setting is far more flexible. Moreover, in locally
cartesian closed models of type theory, predicates and types coexist
in the same category, so that each functor can be taken to be its own
lifting. In the general fibrational setting, predicates are not simply
functions or morphisms, properties and types do not coexist in the
same category, and a functor cannot be taken to be its own
lifting. There is no choice but to construct a lifting from scratch. A
treatment of induction based solely on locally cartesian closed
categories would not, therefore, indicate how to treat induction in
more general fibrations.

Another reason for working in the general fibrational setting is that
this facilitates a direct comparison of our work with that of Hermida
and Jacobs~\cite{hj98}. This is important, since their approach is the
most closely related to ours. The main difference between their
approach and ours is that they use fibred products and coproducts to
define provably sound induction rules for polynomial functors, whereas
we use left adjoints to reindexing functors to define provably sound
induction rules for {\em all} inductive functors. In this section we
consider situations when both approaches are possible and give mild
conditions under which our results coincide with theirs when
restricted to polynomial functors.

The remainder of this section is organised as follows. In
Section~\ref{sec:nutshell} we recall the definition of a fibration,
expand and motivate this definition, and fix some basic terminology
surrounding fibrations. We then give some examples of fibrations,
including the families fibration over $\Set$, the codomain fibration,
and the subobject fibration. In Section~\ref{sec:magictheorem} we
recall a useful theorem from~\cite{hj98} that indicates when a
truth-preserving lifting of a functor to a category of predicates has
an initial algebra. This is the key theorem used to prove the
soundness of our generic fibrational induction rule. In
Section~\ref{sec:liftings} we construct truth-preserving liftings for
all inductive functors. We do this first in the codomain fibration,
and then, using intuitions from its presentation as the families
fibration over $\Set$, as studied in Section~\ref{sec:essence}, in a
general fibrational setting. Finally, in Section~\ref{sec:algebra} we
establish a number of properties of the liftings, and hence of the
induction rules, that we have derived. In particular, we characterise
the lifting that generates our induction rules.

\subsection{Fibrations in a Nutshell}\label{sec:nutshell} 

In this section we recall the notion of a fibration.  More details
about fibrations can be found in, e.g.,~\cite{jac99,pav90}. We begin
with an auxiliary definition.

\begin{defi}
  Let $U : \E \ra \B$ be a functor.  
\begin{enumerate}[(1)]
\item A morphism $g : Q \ra P$ in $\E$ is {\em cartesian} over a
  morphism $f : X \ra Y$ in $\B$ if $Ug = f$, and for every $g' : Q'
  \ra P$ in $\E$ for which $Ug' = f \circ v$ for some $v : UQ' \ra X$
  there exists a unique $h : Q' \ra Q$ in $\E$ such that $Uh = v$ and
  $g \circ h = g'$.
\item A morphism $g : P \ra Q$ in $\E$ is {\em opcartesian} over a
  morphism $f : X \ra Y$ in $\B$ if $Ug = f$, and for every $g' : P
  \ra Q'$ in $\E$ for which $Ug' = v \circ f$ for some $v : Y \ra UQ'$
  there exists a unique $h : Q \ra Q'$ in $\E$ such that $Uh = v$ and
  $h \circ g = g'$.
\end{enumerate}
\end{defi}

\noindent It is not hard to see that the cartesian morphism $f^\S_P$ over a
morphism $f$ with codomain $UP$ is unique up to isomorphism, and
similarly for the opcartesian morphism $f_\S^P$. If $P$ is an object
of $\E$, then we write $f^*P$ for the domain of $f^\S_P$ and
$\Sigma_fP$ for the codomain of $f_\S^P$. We can capture cartesian and
opcartesian morphisms diagrammatically as follows.

\vspace*{0.1in}

\xymatrix @C=1in
{ \E \ar[ddd]_-U & Q' \ar @{.>} [d]_h \ar[rd]^{g'}         &   &
& Q'\\
                        & f^* P \ar[r]_{f^{\S}_P}      & P & P \ar[ru]^{g'}
\ar[r]_{f_{\S}^P} & \Sigma_f P \ar @{.>}[u]_h\\
                        &  UQ' \ar[d]_v \ar[rd]^{U g'}   &   &
& UQ'\\
           \B           & X \ar[r]_f                 & Y   & X \ar[ru]^{Ug'}
\ar[r]_f      & Y \ar[u]_v\\
}

\vspace*{0.1in}

Cartesian morphisms (opcartesian morphisms) are the essence of
fibrations (resp., opfibrations). We introduce both fibrations and
their duals now since the latter will prove useful later in our
development. Below we speak primarily of fibrations, with the
understanding that the dual observations hold for opfibrations.

\begin{defi}
  Let $U : \E \ra \B$ be a functor. Then $U$ is a {\em fibration} if
  for every object $P$ of $\E$, and every morphism $f : X \ra UP$ in
  $\B$ there is a cartesian morphism $f^{\S}_P : Q \ra P$ in $\E$
  above $f$.  Similarly, $U$ is an {\em opfibration} if for every
  object $P$ of $\E$, and every morphism $f : UP \ra Y$ in $\B$ there
  is an opcartesian morphism $f_{\S}^P : P \ra Q$ in $\E$ above $f$.
  A functor $U$ a {\em bifibration} if it is simultaneously a
  fibration and an opfibration.
\end{defi}
If $U : \E \to \B$ is a fibration, we call $\B$ the {\em base
  category} of $U$ and $\E$ the {\em total category} of $U$.  Objects
of the total category $\E$ can be thought of as properties, objects of
the base category $\B$ can be thought of as types, and $U$ can be
thought of as mapping each property $P$ in $\E$ to the type $UP$ of
which $P$ is a property. One fibration $U$ can capture many different
properties of the same type, so $U$ is not injective on objects.  We
say that an object $P$ in $\E$ is \emph{above} its image $UP$ under
$U$, and similarly for morphisms. For any object $X$ of $\B$, we write
$\E_X$ for the {\em fibre above $X$}, i.e., for the subcategory of
$\E$ consisting of objects above $X$ and morphisms above $id$.  If $f
: X \to Y$ is a morphism in $\B$, then the function mapping each
object $P$ of $\E$ to $f^*P$ extends to a functor $f^* : \E_Y \to
\E_X$. Indeed, for each morphism $k : P \to P'$ in $\E_{Y}$, $f^*k$ is
the morphism satisfying $ k \circ f^\S_P = f^\S_{P'} \circ f^* k$.
The universal property of $f^\S_{P'}$ ensures the existence and
uniqueness of $f^* k$. We call the functor $f^*$ the {\em reindexing
  functor induced by $f$}. A similar situation ensures for
opfibrations, and we call the functor $\Sigma_f : \E_X \to \E_Y$ which
extends the function mapping each object $P$ of $\E$ to $\Sigma_f P$
the {\em opreindexing functor}.

\begin{ex}\label{ex:famfib}
  The functor $U:{{\mathcal{P}}}\to\Set$ defined in
  Section~\ref{sec:essence} is called the {\em families fibration over
    $\Set$}. Given a function $f: X \ra Y$ and a predicate $P : Y \ra
  \Set$ we can define a cartesian map $f^{\S}_P$ whose domain $f^*P$
  is $P \circ f$, and which comprises the pair $(f, \lambda x:X.\,
  id)$.  The fibre $\mathcal{P}_X$ above a set $X$ has
  predicates $P:X \ra \Set$ as its objects. A morphism in
  $\mathcal{P}_X$ from $P : X \ra \Set$ to $P' : X \ra \Set$ is a
  function of type $\forall x : X.\, Px \ra P'x$.
\end{ex}

\begin{ex}\label{ex:cod}
  Let $\B$ be a category. The {\em arrow category} of $\B$, denoted
  $\B^\to$, has the morphisms, or arrows, of $\B$ as its objects. A
  morphism in $\B^\to$ from $f:X\to Y$ to $f':X'\to Y'$ is a pair
  $(\alpha_1, \alpha_2)$ of morphisms in $\B$ such that the following
  diagram commutes:
  \[ \xymatrix{ X \;\; \ar[r]^{\alpha_1} \ar[d]_{f} &
    X' \ar[d]^{f'} \\
    Y \ar[r]_{\alpha_2} & Y'}\]
  \noindent
  i.e., such that $\alpha_2 \circ f = f' \circ \alpha_1$. It is easy
  to see that this definition indeed gives a category.

  The codomain functor $\mathit{cod} : \B^{\ra} \ra \B$ maps an object
  $f:X \ra Y$ of $\B^{\ra}$ to the object $Y$ of $\B$ and a morphism
  $(\alpha_1, \alpha_2)$ of $\B^\to$ to $\alpha_2$.  If $\B$ has
  pullbacks, then $\mathit{cod}$ is a fibration, called the {\em
    codomain fibration over $\B$}. Indeed, given an object $f : X \to
  Y$ in the fibre above $Y$ and a morphism $f':X' \ra Y$ in $\B$, the
  pullback of $f$ along $f'$ gives a cartesian morphism above $f'$ as
  required. The fibre above an object $Y$ of $\B$ has those morphisms
  of $\B$ that map into $Y$ as its objects. A morphism in $(\B^\to)_Y$
  from $f:X \to Y$ to $f' : X' \to Y$ is a morphism $\alpha_1 : X \to
  X'$ in $\B$ such that $f = f' \circ \alpha_1$.
\end{ex}

\begin{ex}
  If $\B$ is a category, then the category of subobjects of $\B$,
  denoted $Sub(\B)$, has monomorphisms in $\B$ as its objects. A
  monomorphism $f : X \hookrightarrow Y$ is called a {\em subobject}
  of $Y$.  A morphism in $Sub(\B)$ from $f : X \hookrightarrow Y$ to
  $f' : X' \hookrightarrow Y'$ is a pair of morphisms $(\alpha_1,
  \alpha_2)$ in $\B$ such that $\alpha_2 \circ f = f' \circ \alpha_1$.

  The map $U:Sub(\B) \ra \B$ sending a subobject $f : X
  \hookrightarrow Y$ to $Y$ extends to a functor.  If $\B$ has
  pullbacks, then $U$ is a fibration, called the {\em subobject
    fibration over $\B$}; indeed, pullbacks again give cartesian
  morphisms since the pullback of a monomorphism is a
  monomorphism. The fibre above an object $Y$ of $\B$ has as objects
  the subobjects of $Y$. A morphism in $Sub(\B)_Y$ from $f : X
  \hookrightarrow Y$ to $f': X' \hookrightarrow Y$ is a map $\alpha_1
  : X \to X'$ in $\B$ such that $f = f' \circ \alpha_1$. If such a
  morphism exists then it is, of course, unique.
\end{ex}

\subsection{Lifting, Truth, and Comprehension}\label{sec:magictheorem}

We now generalise the notions of lifting, truth, and comprehension
to the general fibrational setting. We prove that, in such a setting,
if an inductive functor has a truth-preserving lifting, then its
lifting is also inductive. We then see that inductiveness of the
lifted functor is sufficient to guarantee the soundness of our generic
fibrational induction rule. This subsection is essentially our
presentation of pre-existing results from~\cite{hj98}. We include it
because it forms a natural part of our narrative, and because simply
citing the material would hinder the continuity of our presentation.

Recall from Section~\ref{sec:essence} that the first step in deriving
an induction rule for a datatype interpreted in $\Set$ is to lift the
functor whose fixed point the data type is to the category
$\mathcal{P}$ of predicates. More specifically, in
Definition~\ref{def:liftingA} we defined a lifting of a functor $F:
\Set \ra \Set$ to be a functor $\hat{F}:{\mathcal{P}} \ra
{\mathcal{P}}$ such that $U \hat{F} = F U$. We can use these
observations to generalise the notion of a lifting to the fibrational
setting as follows.

\begin{defi}
  Let $U:\E \ra \B$ be a fibration and $F$ be a functor on $\B$. A
  {\em lifting} of $F$ with respect to $U$ is a functor $\hat{F}:\E
  \rightarrow \E$ such that the following diagram commutes:
  \[ \xymatrix{ \E \ar[r]^{\hat{F}} \ar[d]_{U} & \E \ar[d]^U \\
    \B \ar[r]_{F} & \B}
\]
\end{defi}
\noindent
In Section~\ref{sec:essence} we saw that if $P:X \ra \Set$ is a
predicate over $X$, then $\hat{F} P$ is a predicate over $FX$. The
analogous result for the general fibrational setting observes that if
$\hat{F}$ is a lifting of $F$ and $X$ is an object of $\B$, then
$\hat{F}$ restricts to a functor from $\E_X$ to $\E_{F X}$.

By analogy with our results from Section~\ref{sec:essence}, we further
expect that the premises of a fibrational induction rule for a
datatype $\mu F$ interpreted in $\B$ should constitute an
$\hat{F}$-algebra on $\E$. But in order to construct the conclusion of
such a rule, we need to understand how to axiomatically state that a
predicate is true. In Section~\ref{sec:essence}, a predicate $P:X \ra
\Set$ is considered true if there is a morphism in $\mathcal{P}$ from
$K_1X$, the truth predicate on $X$, to $P$ that is over
$\mathit{id}_X$. Since the mapping of each set $X$ to $K_1 X$ is the
action on objects of the truth functor $K_1:\Set \ra {\mathcal{P}}$
(cf.~Definition~\ref{def:truth}), we actually endeavour to model the
truth functor for the families fibration over $\Set$ axiomatically in
the general fibrational setting.

Modeling the truth functor axiomatically amounts to understanding its
universal property. Since the truth functor in
Definition~\ref{def:truth} maps each set $X$ to the predicate $\lambda
x : x.\, 1$, for any set $X$ there is therefore exactly one morphism
in the fibre above $X$ from any predicate $P$ over $X$ to
$K_1\,X$. This gives a clear categorical description of $K_1\,X$ as
a terminal object of the fibre above $X$ and leads, by analogy, to
the following definition.

\begin{defi}\label{def:gtruth}
  Let $U:\E\to\B$ be a fibration. Assume further that, for every
  object $X$ of $\B$, the fibre $\E_X$ has a terminal object $K_1\,X$
  such that, for any $f:X'\to X$ in $\B$, $f^*(K_1\,X) \cong
  K_1\,X'$. Then the assignment sending each object $X$ in $\B$ to
  $K_1 X$ in $\E$, and each morphism $f : X' \to X$ in $\B$ to the
  morphism $f^\S_{K_1X}$ in $\E$ defines the (fibred) \emph{truth
    functor} $K_1 : \B\to\E$.
\end{defi}

\noindent
The (fibred) truth functor is sometimes called the (fibred) {\em
  terminal object functor}. With this definition, we have the
following standard result:

\begin{lem}\label{lem:k1frau}
$K_1$ is a (fibred) right adjoint for $U$.
\end{lem}

The interested reader may wish to consult the literature on fibrations
for the definition of a fibred adjunction, but a formal definition
will not be needed here. Instead, we can simply stress that a fibred
adjunction is first and foremost an adjunction, and then observe that
the counit of this adjunction is the identity, so that $U K_1 =
\mathit{Id}$. Moreover, $K_1$ is full and faithful. One simple way to
guarantee that a fibration has a truth functor is to assume that both
$\E$ and $\B$ have terminal objects and that $U$ maps a terminal
object of $\E$ to a terminal object of $\B$.  In this case, the fact
that reindexing preserves fibred terminal objects ensures that every
fibre of $\E$ indeed has a terminal object.

The second fundamental property of liftings used in
Section~\ref{sec:essence} is that they are truth-preserving. This
property can now easily be generalised to the general fibrational
setting (cf.~Definition~\ref{lem:aux}).

\begin{defi}
  Let $U:\E \ra \B$ be a fibration with a truth functor $K_1:\B \ra
  \E$, let $F$ be a functor on $\B$, and let $\hat{F}:\E \rightarrow
  \E$ be a lifting of $F$. We say that $\hat{F}$ is a {\em
    truth-preserving} lifting of $F$ if, for any object $X$ of $\B$,
  we have $\hat{F} (K_1 X) \cong K_1 (FX)$. 
\end{defi}

The final algebraic structure we required in Section~\ref{sec:essence}
was a comprehension functor $\{-\} : {\mathcal{P}} \ra \Set$. To
generalise the comprehension functor to the general fibrational
setting we simply note that its universal property is that it is right
adjoint to the truth functor $K_1$
(cf.~Definition~\ref{lem:adjoints}). We single out for special
attention those fibrations whose truth functors have right adjoints.

\begin{defi}\label{def:ccu}
  Let $U : \E\to\B$ be a fibration with a truth functor $K_1 :
  \B\to\E$. Then $U$ is a {\em comprehension category with unit} if
  $K_1$ has a right adjoint.
\end{defi}
\noindent
If $U : \E\to\B$ is a comprehension category with unit, then we call
the right adjoint to $K_1$ the {\em comprehension functor} and denote
it by $\{-\}:\E \ra \B$.  With this machinery in place, Hermida and
Jacobs~\cite{hj98} show that if $U$ is a comprehension category with
unit and $\hat{F}$ is a truth-preserving lifting of $F$, then
$\hat{F}$ is inductive if $F$ is and, in this case, the carrier $\mu
\hat{F}$ of the initial $\hat{F}$-algebra is $K_1 (\mu F)$. This is
proved as a corollary to the following more abstract theorem.

\begin{thm}\label{thm:adj}
  Let $F : \B \to \B$, $G : \A \to \A$, and $S : \B \to \A$ be
  functors. A natural transformation $\alpha : G S \to S F$, i.e., a
  natural transformation $\alpha$ such that
 \[\dia{
   {\A} \ar[r]^G \ar@{=>}"1,1"+<14px,-14px>;"2,2"-<14px,-14px>^\alpha & {\A}  \\
   \B \ar[u]^S \ar[r]_F & {\B} \ar[u]_S
 }\]
\noindent
induces a functor 

\[\dia{Alg_F \ar[r]^\Phi &Alg_G}\]

\noindent
given by $\Phi\ (f : F X \to X) = S\,f\circ \alpha_X$. Moreover, if
$\alpha$ is an isomorphism, then a right adjoint $T$ to $S$ induces a
right adjoint 
\[\adjunction{Alg_F}{Alg_G}{\Phi}{\Psi}\]

\noindent
given by $\Psi (g : GX \ra X) = Tg \circ \beta_X$, where $\beta :
F T \to T G$ is the image of $G\,\epsilon \circ \alpha^{-1}_T
: S F T \to G$ under the adjunction isomorphism
$Hom(S\,X,\,Y)\cong Hom(X,\,T\,Y)$, and $\epsilon : S T \to
id$ is the counit of this adjunction.
\end{thm}

\noindent
We can instantiate Theorem~\ref{thm:adj} to generalise
Lemmas~\ref{lem:Phi} and~\ref{lem:Psi}.

\begin{thm}\label{thm:genpsiphi}
  Let $U:\E\to\B$ be a comprehension category with unit and
  $F:\B\to\B$ be a functor. If $F$ has a truth-preserving lifting
  $\hat{F}$ then there is an adjunction $\Phi \dashv \Psi:\Alge{F}
  \rightarrow \Alge{\hat{F}}$. Moreover, if $f: FX \ra X$ then $\Phi f
  : \hat{F}(K_1X) \ra K_1X$, and if $g : \hat{F}P \ra P$ then $\Psi g
  : F\{P\} \ra \{P\}$.
\end{thm}
\begin{proof}
  We instantiate Theorem~\ref{thm:adj}, letting $\E$ be $\A$,
  $\hat{F}$ be $G$, and $K_1$ be $S$. Then $\alpha$ is an isomorphism
  since $\hat{F}$ is truth-preserving, and we also have that $K_1
  \dashv \{-\}$. The theorem thus ensures that $\Phi$ maps every
  $F$-algebra $f : FX \ra X$ to an $\hat{F}$-algebra $\Phi f :
  \hat{F}(K_1X) \ra K_1X$, and that $\Psi$ maps every
  $\hat{F}$-algebra $g : \hat{F}P \ra P$ to an $F$-algebra $\Psi g : 
  F\{P\} \ra \{P\}$.
\end{proof}

\begin{cor}
  Let $U:\E\to\B$ be a comprehension category with unit and
  $F:\B\to\B$ be a functor which has a truth-preserving lifting
  $\hat{F}$. If $F$ is inductive, then so is $\hat{F}$. Moreover, $\mu
  \hat{F} = K_1(\mu F)$.
\end{cor}
\begin{proof}
  The hypotheses of the corollary place us in the setting of
  Theorem~\ref{thm:genpsiphi}. This theorem guarantees that $\Phi$
  maps the initial $F$-algebra $\mathit{in}_F : F(\mu F) \ra \mu F$ to
  an $\hat{F}$-algebra with carrier $K_1(\mu F)$.  But since left
  adjoints preserve initial objects, we must therefore have that the
  initial $\hat{F}$-algebra has carrier $K_1(\mu F)$. Thus, $\mu
  \hat{F}$ exists and is isomorphic to $K_1(\mu F)$.
\end{proof}

\begin{thm}\label{thm:induction}
  Let $U:\E\to\B$ be a comprehension category with unit and
  $F:\B\to\B$ be an inductive functor. If $F$ has a truth-preserving
  lifting $\hat{F}$, then the following generic fibrational induction
  rule is sound:
\[\begin{array}{lll}
  \mathit{genfibind} & : & \forall \, (F : \B \rightarrow \B) \; (P
  :\E). \; (\hat{F}\,P \ra P) \to (\mu \hat{F} \ra P)\\
  \mathit{genfibind} \, F \, P & = & \mathit{fold} 
\end{array}\]
\end{thm}
\noindent
An alternative presentation of $\mathit{genfibind}$ is
\[\begin{array}{lll}
  \mathit{genfibind} & : & \forall \, (F : \B \rightarrow \B) \; (P
  :\E). \; (\hat{F}\,P \ra P) \to (\mu F \ra \{P\})\\
  \mathit{genfibind} \, F \, P & = & \mathit{fold} \circ \Psi
\end{array}\]
\noindent
We call $\mathit{genfibind}\,F$ the {\em generic fibrational induction
  rule} for $\mu F$.

In summary, we have generalised the generic induction rule for
predicates over $\Set$ presented in Section~\ref{sec:essence} to give
a sound generic induction rule for comprehension categories with
unit. Our only assumption is that if we start with an inductive
functor $F$ on the base of the comprehension category, then there must
be a truth-preserving lifting of that functor to the total category of
the comprehension category. In that case, we can specialise
$\mathit{genfibind}$ to get a fibrational induction rule for any
datatype $\mu F$ that can be interpreted in the fibration's base
category.

The generic fibrational induction rule $\mathit{genfibind}$ does,
however, look slightly different from the generic induction rule for
set-valued predicates. This is because, in Section~\ref{sec:essence},
we used our knowledge of the specific structure of comprehensions for
set-valued predicates to extract proofs for particular data elements
from them. But in the fibrational setting, predicates, and hence
comprehensions, are left abstract. We therefore take the return type
of the general induction scheme $\mathit{genfibind}$ to be a
comprehension with the expectation that, when the general theory of
this section is instantiated to a particular fibration of interest, it
may be possible to use knowledge about that fibration to extract from
the comprehension constructed by $\mathit{genfibind}$ further proof
information relevant to the application at hand. 

As we have previously mentioned, Hermida and Jacobs provide
truth-preserving liftings only for polynomial functors. In
Section~\ref{sec:liftings}, we define a generic truth-preserving
lifting for \emph{any} inductive functor on the base category of any
fibration which, in addition to being a comprehension category with
unit, has left adjoints to all reindexing functors. This gives a sound
generic fibrational induction rule for the datatype $\mu F$ for any
functor $F$ on the base category of any such fibration.

\subsection{Constructing Truth-Preserving Liftings}\label{sec:liftings}

In light of the previous subsection, it is natural to ask whether or
not truth-preserving liftings exist. If so, are they unique? Or, if
there are many truth-preserving liftings, is there a specific
truth-preserving lifting to choose above others? Is there, perhaps,
even a universal truth-preserving lifting? We can also ask about the
algebraic structure of liftings. For example, do truth-preserving
liftings preserve sums and products of functors?

Answers to some of these questions were given by Hermida and Jacobs,
who provided truth-preserving liftings for polynomial functors. To
define such liftings they assume that the total category and the base
category of the fibration in question have products and coproducts,
and that the fibration preserves them. Under these conditions,
liftings for polynomial functors can be defined inductively. In this
section we go beyond the results of Hermida and Jacobs and construct
truth-preserving liftings for {\em all} inductive functors. We employ
a two-stage process, first building truth-preserving liftings under
the assumption that the fibration of interest is a codomain fibration,
and then using the intuitions of Section~\ref{sec:essence} to extend
this lifting to a more general class of fibrations. In
Section~\ref{sec:algebra} we consider the questions from the previous
paragraph about the algebraic structure of liftings.

\subsubsection{Truth-Preserving Liftings for Codomain Fibrations}

Recall from Example~\ref{ex:cod} that if $\B$ has pullbacks, then the
codomain fibration over $\B$ is the functor $\mathit{cod}:\B^\ra \ra
\B$.  Given a functor $F:\B \ra \B$, it is trivial to define a lifting
$F^\ra:\B^\ra \ra \B^\ra$ for this fibration. We can define the
functor $F^\ra$ to map an object $f:X \ra Y$ of $\B^\ra$ to $Ff: FX
\ra FY$, and to map a morphism $(\alpha_1,\alpha_2)$ to the morphism
$(F \alpha_1, F \alpha_2)$. That $F^\ra$ is a lifting is easily
verified.

If we further verify that codomain fibrations are comprehension
categories with unit, and that the lifting $F^\ra$ is
truth-preserving, then Theorem~\ref{thm:induction} can be applied to
them. For the former, we first observe that the functor $K_1 : \B \ra
\B^\ra$ mapping an object $X$ to $\mathit{id}$ and a morphism $f : X
\ra Y$ to $(f,f)$ is a truth functor for this fibration. (In fact, we
can take any isomorphism into $X$ as $K_1X$; we will use this
observation below.) If we let $\B^\to(U,V)$ denote the set of
morphisms from an object $U$ to an object $V$ in $\B^\to$, then the
fact that $K_1$ is right adjoint to $\mathit{cod}$ can be established
via the natural isomorphism
\[\B^\ra(f:X \ra Y,\, K_1Z) = \{ (\alpha_1 : X \ra Z,\, \alpha_2 : Y
\ra
Z) \,| \,\alpha_1 = \alpha_2\, \circ\, f \} \cong \B(Y, Z) = \B
(\mathit{cod}\, f,\, Z)\]

We next show that the functor $\mathit{dom}:\B^\ra \ra \B$ mapping an
object $f: X \ra Y$ of $\B^\ra$ to $X$ and a morphism $(\alpha_1,
\alpha_2)$ to $\alpha_1$ is a comprehension functor for the codomain
fibration. That $\mathit{dom}$ is right adjoint to $K_1$ is
established via the natural isomorphism
\[\B^\ra(K_1Z,\, f: X \ra Y) = \{ (\alpha_1 : Z \ra X, \alpha_2 : Z
\ra Y)\, |\, \alpha_2 = f\, \circ\, \alpha_1\} \cong \B(Z, X) = \B (Z,
\mathit{dom} \, f) \]

Finally, we have that $F^\ra$ is truth-preserving because
\[F^\ra (K_1 Z) = F^\ra\, \mathit{id} = F\, \mathit{id} = \mathit{id}
= K_1 (F Z)\]

A lifting is implicitly given in~\cite{men91} for functors on a
category with display maps. Such a category is a subfibration of the
codomain fibration over that category, and the lifting given there is
essentially the lifting for the codomain fibration restricted to the
subfibration in question.

\subsubsection{Truth-Preserving Liftings for the Families Fibration
  over $\Set$}\label{sec:fam}

In Section~\ref{sec:essence} we defined, for every functor $F:\Set \ra
\Set$, a lifting $\hat{F}$ which maps the predicate $P$ to $(F
\pi_P)^{-1}$. Looking closely, we realise this lifting decomposes into
three parts. Given a predicate $P$, we first consider the projection
function $\pi_P : \{P\} \ra UP$. Next, we apply the functor $F$ to
$\pi_P$ to obtain $F \pi_P : F \{P\} \ra FUP$. Finally, we take the
inverse image of $F\pi_P$ to get a predicate over $FUP$ as required.

Note that $\pi$ is the functor from $\mathcal{P}$ to $\Set^\to$ which
maps a predicate $P$ to the projection function $\pi_P : \{P\} \ra UP$
(and maps a predicate morphism $(f,f^\sim)$ from a predicate $P : X
\ra \Set$ to $P' : X' \ra \Set$ to the morphism $(\{(f,f^\sim)\}, \;
f)$ from $\pi_P$ to $\pi_{P'}$;
cf.~Definition~\ref{def:comprehension}). If $I:\Set^\ra \ra
\mathcal{P}$ is the functor sending a function $f:X \ra Y$ to its
``inverse'' predicate $f^{-1}$ (and a morphism $(\alpha_1,\alpha_2)$
to the predicate morphism $(\alpha_2, \forall y:Y.\; \lambda x:
f^{-1}y. \; \alpha_1 x)$), then each of the three steps of defining
$\hat{F}$ is functorial and the relationships indicated by the
following diagram hold:
\[\xymatrix{{\mathcal{P}}\ar[dr]_U \ar@/^/[rr]^\pi \ar@{}[rr]|\top & &
  \ar@/^/[ll]^I \ar[ld]^{cod} {\Set^\to}\\ & {\Set} & }\]

\vspace*{0.1in}

Note that the adjunction $I \dashv \pi$ is an equivalence. This
observation is not, however, necessary for our subsequent development;
in particular, it is not needed for Theorem~\ref{thm:tp}.

The above presentation of the lifting $\hat{F}$ of a functor $F$ for
the families fibration over $\Set$ uses the lifting of $F$ for the
codomain fibration over $\Set$. Indeed, writing $F^\to$ for the
lifting of $F$ for the codomain fibration over $\Set$, we have that
$\hat{F} = I F^\to \pi$. Moreover, since $\pi$ and $I$ are
truth-preserving (see the proof of Lemma~\ref{lem:aux}), and since we
have already seen that liftings for codomain fibrations are
truth-preserving, we have that $\hat{F}$ is truth-preserving because
each of its three constituent functors is. Finally, since we showed in
Section~\ref{sec:essence} that the families fibration over $\Set$ is a
comprehension category with unit, Theorem~\ref{thm:induction} can be
applied to it.

Excitingly, as we shall see in the next subsection, the above
presentation of the lifting of a functor for the families fibration
over $\Set$ generalises to many other fibrations!

\subsubsection{Truth-Preserving Liftings for Other
  Fibrations}\label{sec:general}

We now turn our attention to the task of constructing truth-preserving
liftings for fibrations other than codomain fibrations and the
families fibration over $\Set$. By contrast with the approach outlined
in the conference paper~\cite{gjf10} on which this paper is based, the
one we take here uses a factorisation, like that of the previous
subsection, through a codomain fibration. More specifically, let $U:\E
\ra \B$ be a comprehension category with unit. We first define
functors $I$ and $\pi$, and construct an adjunction $I \dashv \pi$
between $\E$ and $\B^\ra$ such that the relationships indicated by the
following diagram hold:
\[\xymatrix{{\E}\ar[dr]_U \ar@/^/[rr]^\pi \ar@{}[rr]|\top & &
   \ar@/^/[ll]^I \ar[ld]^{cod} {\B^\to}\\ & {\B} & }\]
\noindent
We then use the adjunction indicated in the diagram to construct
truth-preserving a lifting for $U$ from that for the codomain
fibration over $\B$.

To define the functor $\pi:\E \ra \B^\ra$ we generalise the definition
of $\pi : \mathcal{P} \to \Set^\to$ from Sections~\ref{sec:essence}
and~\ref{sec:fam}. This requires us to work with the axiomatic
characterisation in Definition~\ref{def:ccu} of the comprehension
functor $\{-\} : \E \to \B$ as the right adjoint to the truth functor
$K_1 : \B \to \E$.  The counit of the adjunction $K_1 \dashv \{-\}$ is
a natural transformation $\epsilon:K_1 \{-\} \ra Id$. Applying $U$ to
$\epsilon$ gives the natural transformation $U \epsilon : U K_1 \{-\}
\ra U$, but since $U K_1 = Id$, in fact we have that $U \epsilon :
\{-\} \ra U$. We can therefore define $\pi$ to be $U \epsilon$. Then
$\pi$ is indeed a functor from $\E$ to $\B^\to$, its action on an
object $P$ is $\pi_P$, and its action on a morphism $(f,f^\sim)$ is
$(\{(f,f^\sim)\},\,f)$.

We next turn to the definition of the left adjoint $I$ to $\pi$.  To
see how to generalise the inverse image construction to more general
fibrations we first recall from Example~\ref{ex:famfib} that, if $f :
X \to Y$ is a function and $P : Y \to \Set$, then $f^* P = P \circ
f$. We can extend this mapping to a reindexing functor
$f^*:\E_Y\to\E_X$ by defining $f^* (\mathit{id},h^\sim) =
(\mathit{id}, h^\sim \circ f)$. If we define the action of $\Sigma_f
: \E_X \ra \E_Y$ on objects by
\[\Sigma_f P \;=\; \lambda y.\ \biguplus_{\{x|f\,x = y\}} P\,x \]
where $\biguplus$ denotes the disjoint union operator on sets, and its
action on morphisms by taking $\Sigma_f\, (\mathit{id}, \alpha^\sim)$
to be $(\mathit{id},\; \forall (y:Y).\; \lambda (x:X,p:fx=y,t:Px).\,
(x,p,\alpha^\sim \,x \,t))$, then $\Sigma_f$ is left adjoint to
$f^*$. Moreover, if we compute
\[\Sigma_f\, (K_1\,X) \;=\; \lambda y.\ \biguplus_{\{x\,|\,f\,x = y\}} K_1Xx\]
and recall that, for any $x:X$, the set $K_1\,X\,x$ is a singleton,
then $\Sigma_f \, (K_1X)$ is clearly equivalent to the inverse image
of $f$.

The above discussion suggests that, in order to generalise the inverse
image construction to a more general fibration $U : \E \to \B$, we
should require each reindexing functor $f^*$ to have the opreindexing
functor $\Sigma_f$ as its left adjoint.  As in~\cite{hj98}, no
Beck-Chevalley condition is required on these adjoints. The following
result, which appears as Proposition~2.3 of~\cite{jac93}, thus allows
us to isolate the exact class of fibrations for which we will have
sound generic induction rules.

\begin{thm}
  A fibration $U : \E \to \B$ is a bifibration iff for every morphism
  $f$ in $\B$ the reindexing functor $f^*$ has left adjoint
  $\Sigma_f$.
\end{thm}

\begin{defi}\label{def:lawvere}
  A {\em Lawvere category} is a bifibration which is also a
  comprehension category with unit.
\end{defi}

We construct the left adjoint $I:\B^\to \to \E$ of $\pi$ for any
Lawvere category $U:\E \ra \B$ as follows. If $f : X \to Y$ is an
object of $\B^\to$, i.e., a morphism of $\B$, then we define $I\,f$ to
be the object $\Sigma_f (K_1 X)$ of $\E$. To define the action of $I$
on morphisms, let $(\alpha_1, \alpha_2)$ be a morphism in $\B^\to$
from $f:X\to Y$ to $f':X'\to Y'$ in $\B^\to$. Then $(\alpha_1,
\alpha_2)$ is a pair of morphisms in $\B$ such that the following
diagram commutes:
\[ \xymatrix{ X \;\; \ar[r]^{\alpha_1} \ar[d]_{f} &
    X' \ar[d]^{f'} \\
    Y \ar[r]_{\alpha_2} & Y'}
\]
We must construct a morphism from $\Sigma_f (K_1 X)$ to
$\Sigma_{f'}(K_1 X')$ in $\E$. To do this, notice that
$f'^{K_1X'}_\S\circ K_1 \alpha_1 : K_1 X\to \Sigma_{f'} (K_1X')$ is
above $f'\circ \alpha_1$, and that it is also above $\alpha_2\circ f$
since $f'\circ \alpha_1=\alpha_2\circ f$. We can then consider the
morphism $f'^{K_1X'}_\S\circ K_1 \alpha_1$ and use the universal
property of the opcartesian morphism $f^{K_1 X}_\S$ to deduce the
existence of a morphism $h:\Sigma_f (K_1 X)\to\Sigma_{f'}(K_1 X')$
above $\alpha_2$.  It is not difficult, using the uniqueness of the
morphism $h$, to prove that setting this $h$ to be the image of the
morphism $(\alpha_1,\alpha_2)$ makes $I$ a functor.  In fact, since
$\mathit{cod} \circ \pi = U$, Result (i) on page 190 of~\cite{jac93}
guarantees that, for any Lawvere category $U:\E\to\B$ the functor
$I:\B^\to\to\E$ exists and is left adjoint to $\pi:\E\to\B^\to$.

\vspace*{0.1in}

We can now construct a truth-preserving lifting for any Lawvere
category $U:\E \ra \B$ and functor $F$ on $\B$.

\begin{thm}\label{thm:tp}
  Let $U:\E\to\B$ be a Lawvere category and, for any functor $F$ on $\B$,
  define the functor $\hat{F}$ on $\E$ by
\[\begin{array}{lll}
\hat{F} & : & \E\to\E\\
\hat{F} & = & IF^\to\pi
\end{array}\]
Then $\hat{F}$ is a truth-preserving lifting of $F$.
\end{thm}

\begin{proof}
  It is trivial to check that $\hat{F}$ is indeed a lifting. To prove
  that it is truth-preserving, we need to prove that $\hat{F}(K_1 X)
  \cong K_1 (F X)$ for any functor $F$ on $\B$ and object $X$ of
  $\B$. We do this by showing that each of $\pi$, $F^\to$, and $I$
  preserves fibred terminal objects, i.e., preserves the terminal
  objects of each fibre of the total category which is its domain.
  Then since $K_1X$ is a terminal object in the fibre $\E_X$, we will
  have that $\hat{F}(K_1X) = I(F^\to(\pi(K_1X)))$ is a terminal object
  in $\E_{FX}$, i.e., that $\hat{F}(K_1X) \cong K_1(FX)$ as desired.

  We first show that $\pi$ preserves fibred terminal objects. We must
  show that, for any object $X$ of $\B$, $\pi_{K_1 X}$ is a terminal
  object of the fibre of $\B^\to$ over $X$, i.e., is an isomorphism
  with codomain $X$. We prove this by observing that, if $\eta :
  \mathit{Id} \to \{-\}K_1$ is the unit of the adjunction $K_1 \dashv
  \{-\}$, then $\pi_{K_1X}$ is an isomorphism with inverse
  $\eta_X$. Indeed, if $\epsilon$ is the counit of the same
  adjunction, then the facts that $UK_1 = \mathit{Id}$ and that $K_1$
  is full and faithful ensure that $K_1\eta_X$ is an isomorphism with
  inverse $\epsilon_{K_1X}$. Thus, $\epsilon_{K_1X}$ is an isomorphism
  with inverse $K_1\eta_X$, and so $\pi_{K_1X} = U \epsilon_{K_1X}$ is
  an isomorphism with inverse $UK_1\eta_X$, i.e., with inverse
  $\eta_X$. Since $K_1X$ is a terminal object in $\E_X$ and
  $\pi_{K_1X}$ is a terminal object in the fibre of $\B^\to$ over $X$,
  we have that $\pi$ preserves fibred terminal objects.

  It is not hard to see that $F^\ra$ preserves fibred terminal
  objects: applying the functor $F$ to an isomorphism with codomain
  $X$ --- i.e., to a terminal object in the fibre of $\B^\to$ over $X$
  --- gives an isomorphism with codomain $FX$ --- i.e., a terminal
  object in the fibre of $\B^\to$ over $FX$.
 
  Finally, if $f: X \ra Y$ is an isomorphism in $\B$, then $\Sigma_f$
  is not only left adjoint to $f^*$, but also right adjoint to
  it. Since right adjoints preserve terminal objects, and since $K_1
  X$ is a terminal object of $\E_X$, we have that $I f = \Sigma_f (K_1
  X)$ is a terminal object of $\E_Y$. Thus $I$ preserves fibred
  terminal objects.
\end{proof}

We stress that, to define our lifting, the codomain functor over the
base $\B$ of the Lawvere category need not be a fibration. In
particular, $\B$ need not have pullbacks; indeed, all that is needed
to construct our generic truth-preserving lifting $\hat{F}$ for a
functor $F$ on $\B$ is the existence of the functors $I$ and $\pi$
(and $F^\to$, which always exists). We nevertheless present the
lifting $\hat{F}$ as the composition of $\pi$, $F^\to$, and $I$
because this presentation shows it can be factored through $F^\to$.
This helps motivate our definition of $\hat{F}$, thereby revealing
parallels between it and $F^\to$ that would otherwise not be
apparent. At the same time it trades the direct, brute-force
presentation of $\hat{F}$ from~\cite{gjf10} for an elegant modularly
structured one which makes good use, in a different setting, of
general results about comprehension categories due to
Jacobs~\cite{jac93}.

We now have the promised sound generic fibrational induction rule for
every inductive functor $F$ on the base of a Lawvere category. To
demonstrate the flexibility of this rule, we now derive an induction
rule for a data type and properties on it that cannot be modelled in
$\Set$. Being able to derive induction rules for fixed points of
functors in categories other than $\Set$ is a key motivation for
working in a general fibrational setting.

\begin{ex}\label{ex:hyperfunctions}
  The fixed point $Hyp = \mu F$ of the functor $FX = (X \rightarrow
  Int) \rightarrow Int$ is the data type of hyperfunctions. Since $F$
  has no fixed point in $\Set$, we interpret it in the category
  $\omega CPO_{\bot}$ of $\omega$-cpos with $\bot$ and strict
  continuous monotone functions. In this setting, a property of an
  object $X$ of $\omega CPO_{\bot}$ is an admissible sub-$\omega
  CPO_{\bot}$ $P$ of $X$. Admissibility means that the bottom element
  of $X$ is in $P$ and $P$ is closed under least upper bounds of
  $\omega$-chains in $X$. This structure forms a Lawvere
  category~\cite{jac93,jac99}; in particular, it is routine to verify
  the existence of its opreindexing functor. In particular, $\Sigma_f
  P$ is constructed for a continuous map $f:X \to Y$ and an admisible
  predicate $P \subseteq X$, as the intersection of all admissible $Q
  \subseteq Y$ with $P \subseteq f^{-1}(Q)$. The truth functor maps
  $X$ to $X$, and comprehension maps a sub-$\omega CPO_{\bot}$ $P$ of
  $X$ to $P$. The lifting $\hat{F}$ maps a sub-$\omega CPO_{\bot}$ $P$
  of $X$ to the least admissible predicate on $FX$ containing the
  image of $FP$. Finally, the derived induction rule states that if
  $P$ is an admissible sub-$\omega CPO_{\bot}$ of $Hyp$, and if
  $\hat{F}(P) \subseteq P$, then $P = Hyp$.
\end{ex}

\subsection{An Algebra of Lifting}\label{sec:algebra}

We have proved that in any Lawvere category $U:\E \ra \B$, any functor
$F$ on $\B$ has a lifting $\hat{F}$ on $\E$ which is truth-preserving,
and thus has the following associated sound generic fibrational
induction rule: 
\[\begin{array}{lll}
  \mathit{genfibind} & : & \forall \, (F : \B \rightarrow \B) \; (P
  :\E). \; (\hat{F}\,P \ra P) \to (\mu F \ra \{P\})\\
  \mathit{genfibind} \, F \, P & = & \mathit{fold} \circ \Psi 
\end{array}\]
In this final subsection of the paper, we ask what kinds of algebraic
properties the lifting operation has. Our first result concerns the
lifting of constant functors.

\begin{lem}
  Let $U:\E \ra \B$ be a Lawvere category and let $X$ be an object of
  $\B$. If $F_X$ is the constantly $X$-valued functor on $\B$, then
  $\widehat{F_X}$ is isomorphic to the constantly $K_1 X$-valued
  functor on $\E$.
\end{lem}

\begin{proof} For any object $P$ of $\E$ we have
  \[ \widehat{F_X} P = (I (F_X)^\to \pi)P = I(F_X\pi_P) = \Sigma_{F_X
    \pi_P} K_1 F_X \{P\} = \Sigma_{\mathit{id}} K_1 X \cong K_1 X\]
  The last isomorphism holds because $\mathit{id}^* \cong \mathit{Id}$
  and $\Sigma_{\mathit{id}} \dashv \mathit{id}^*$.
\end{proof}

Our next result concerns the lifting of the identity functor. It
requires a little additional structure on the Lawvere category of
interest.

\begin{defi}
  A {\em full Lawvere category} is a Lawvere category $U:\E \ra \B$ such
  that $\pi:\E \ra \B^\ra$ is full and faithful.
\end{defi}

\begin{lem}
  In any full Lawvere category, $\widehat{\mathit{Id}} \cong
  \mathit{Id}$
\end{lem}

\begin{proof}
  By the discussion following Definition~\ref{def:lawvere}, $I \dashv
  \pi$. Since $\pi$ is full and faithful, the counit of this
  adjunction is an isomorphism, and so $I \pi_P \cong P$ for all $P$
  in $\E$. We therefore have that
  \[P \cong I\pi_P = \Sigma_{\pi_p} K_1 \{P\} = \Sigma_{\mathit{Id}\,
    \pi_p} K_1 (\mathit{Id}\, \{P\}) = (I\, \mathit{Id}^\to\, \pi) P =
  \widehat{\mathit{Id}}\,P\]
i.e., that $\widehat{\mathit{Id}}\, P \cong P$ for all $P$ in
$\E$. Because these isomorphisms are clearly natural, we therefore
have that $\widehat{\mathit{Id}} \cong \mathit{Id}$.
\end{proof}

We now show that the lifting of a coproduct of functors is the
coproduct of the liftings.

\begin{lem}\label{lem:sums}
  Let $U:\E \ra \B$ be a Lawvere category and let $F$ and $G$ be
  functors on $\B$. Then $\widehat{F+G} \cong \hat{F} + \hat{G}$.
\end{lem}

\begin{proof}
We have
\[ (\widehat{F+G}) P = I ((F+G)^\ra \pi_P) = I (F^\ra \pi_P + G^\ra
\pi_P) \cong I (F^\ra \pi_P) + I(G^\ra \pi_P) = \hat{F}P + \hat{G} P\]
The third isomorphism holds because $I$ is a left adjoint and so
preserves coproducts.
\end{proof}

\noindent
Note that the statement of Lemma~\ref{lem:sums} does not assert the
existence of either of the two coproducts mentioned, but rather that,
whenever both do exist, they must be equal. Note also that the lemma
generalises to any colimit of functors. Unfortunately, no result
analogous to Lemma~\ref{lem:sums} can yet be stated for products.

Our final result considers whether or not there is anything
fundamentally special about the lifting we have constructed. It is
clearly the ``right'' lifting in some sense because it gives the
expected induction rules. But other truth-preserving liftings might
also exist and, if this is the case, then we might hope our lifting
satisfies some universal property. In fact, under a further condition,
which is also satisfied by all of the liftings of Hermida and Jacobs,
and which we therefore regard as reasonable, we can show that our
lifting is the only truth-preserving lifting. Our proof uses a line of
reasoning which appears in Remark~2.13 in~\cite{hj98}.

\begin{lem}\label{lem:solo}
  Let $U:\E \ra \B$ be a full Lawvere category and let $\Box F$ be a
  truth-preserving lifting of a functor $F$ on $\B$. If $\Box F$
  preserves $\Sigma$-types --- i.e., if $(\Box F)(\Sigma_f P) \cong
  \Sigma_{F f} (\Box F) P$ --- then $\Box F \cong \hat{F}$.
\end{lem}

\proof
  We have
\begin{eqnarray*}
(\Box F) P & \cong & (\Box F) (\widehat{\mathit{Id}} P) \\
           & \cong & (\Box F) (\Sigma_{\pi_P} K_1 \{P\})\\
           & \cong & \Sigma_{F \pi_P} (\Box F) K_1 \{P\} \\
           & \cong & \Sigma_{F \pi_P} K_1 F \{P\} \\
           & =     & \hat{F} P\rlap{\hbox to 212 pt{\hfill\qEd}}
\end{eqnarray*}

\noindent Finally, we can return to the question of the relationship between the
liftings of polynomial functors given by Hermida and Jacobs and the
liftings derived by our methods. We have seen that for constant
functors, the identity functor, and coproducts of functors our
constructions agree. Moreover, since Hermida and Jacobs' liftings all
preserve $\Sigma$-types, Lemma~\ref{lem:solo} guarantees that in a
full Lawvere category their lifting for products also coincides with
ours.

\section{Conclusion and future work}\label{sec:conclusion}

We have given a sound induction rule that can be used to prove
properties of data structures of inductive types. Like Hermida and
Jacobs, we give a fibrational account of induction, but we derive,
under slightly different assumptions on fibrations, a generic
induction rule that can be instantiated to {\em any} inductive type
rather than just to polynomial ones. This rule is based on initial
algebra semantics of data types, and is parameterised over both the
data types and the properties involved. It is also principled,
expressive, and correct. Our derivation yields the same induction
rules as Hermida and Jacobs' when specialised to polynomial functors
in the families fibration over $\Set$ and in other fibrations, but it
also gives induction rules for non-polynomial data types such as rose
trees, and for data types such as finite hereditary sets and
hyperfunctions, for which no fibrational induction rules have
previously been known to exist.

There are several directions for future work. The most immediate is to
instantiate our theory to give induction rules for more sophisticated
data types, such as nested types. These are exemplified by the data
type of perfect trees given in Haskell-like syntax as follows:
\[\begin{array}{l}
\mathit{data \; PTree}\; a : \Set \;\mathit{where}\\
\;\;\; \mathit{PLeaf} : a \ra \mathit{PTree}\; a\\
\;\;\; \mathit{PNode} : \mathit{PTree}\; (a,a) \ra \mathit{PTree}\; a
\end{array}\]
Nested types arise as least fixed points of rank-2 functors; for
example, the type of perfect trees is $\mu H$ for the 
functor $H$ given by $H F = \lambda X.\ X + F (X \times X)$.  An
appropriate fibration for induction rules for nested types thus
takes $\B$ to be the category of functors on $\Set$, $\E$ to
be the category of functors from $\Set$ to ${\mathcal{P}}$, and $U$ to be
postcomposition with the forgetful functor from
Section~\ref{sec:essence}.  A lifting $\hat{H}$ of $H$ is given
by $\hat{H}\, P\, X\, (\mathit{inl}\, a) = 1$ and $\hat{H}\, P\,X\,
(\mathit{inr} \, n) = P \,(X \times X)\, n$.  Taking the premise to be
an $\hat{H}$-algebra gives the following induction rule for perfect
trees:
\begin{align*}
indPTree : \;
&\forall\; (P : \Set \rightarrow {\mathcal{P}}).\\
& (U P = \mathit{PTree})
\rightarrow (\forall (X : \Set) (x : X).\ P \,(\mathit{PLeaf} \; x))
\rightarrow\\ 
&(\forall (X : \Set) (t : \mathit{PTree} \,(X \times X).\, P\, (X
\times X)\, t \rightarrow P \,(\mathit{PNode} \, t))) \rightarrow\\
&\qquad\forall (X : \Set) (t : \mathit{PTree}\, X).\, P\, X\, t
\end{align*}
This rule can be used to show, for example, that $PTree$ is a functor.

Extending the above instantiation for the codomain fibration to
so-called ``truly nested types''~\cite{mat09} and fibrations is
current work. We expect to be able to instantiate our theory for truly
nested types, GADTs, indexed containers, dependent types, and
inductive recursive types, but initial investigations show care is
needed. We must ascertain which fibrations can model predicates on
such types, since the codomain fibration may not give useful induction
rules, as well as how to translate the rules to which these fibrations
give rise to an intensional setting.

Matthes~\cite{mat09} gives induction rules for nested types (including
truly nested ones) in an intensional type theory. He handles only
rank-2 functors that underlie nested types (while we handle any
functor of any rank with an initial algebra), but his insights may
help guide choices of fibrations for truly nested types. These may in
turn inform choices for GADTs, indexed containers, and dependent
types. 

Induction rules can automatically be generated in many type
theories. Within the Calculus of Constructions~\cite{ch88} an
induction rule for a data type can be generated solely from the
inductive structure of that type. Such generation is also a key idea
in the Coq proof assistant~\cite{coq}. As far as we know, generation
can currently be done only for syntactic classes of functors rather
than for all inductive functors with initial algebras.  In some type
theories induction schemes are added as axioms rather than
generated. For example, attempts to generate induction schemes based
on Church encodings in the Calculus of Constructions proved
unsuccessful and so initiality was added to the system, thus giving
the Calculus of Inductive Constructions. Whereas Matthes' work is
based on concepts such as impredicativity and induction recursion
rather than initial algebras, ours reduces induction to initiality,
and may therefore help lay the groundwork for extending
implementations of induction to more sophisticated data types.

\section*{Acknowledgement}

\noindent
We thank Robert Atkey, Pierre-Evariste Dagand, Peter Hancock, and
Conor McBride for many fruitful discussions.

\end{document}